\documentclass[a4paper,10pt]{article}
\usepackage{a4wide}
\usepackage[utf8]{inputenc}
\usepackage[T1]{fontenc}
\usepackage[utf8]{inputenc}
\usepackage[english]{babel}
\usepackage{fancyhdr}
\usepackage{setspace}
\usepackage{subcaption}
\usepackage{amsmath,amssymb}
\usepackage{amsthm}
\usepackage{enumerate}

\usepackage{float}

\usepackage{authblk}

\usepackage[svgnames]{xcolor}
\usepackage[pdftex]{hyperref}
\hypersetup{
    colorlinks=true,
    colorlinks,
    linkcolor=Navy,
    citecolor=Navy,
    urlcolor=Navy
}

\usepackage{graphicx}
\graphicspath{{fig/}}
\usepackage{epstopdf}
\usepackage{booktabs}
\usepackage{multirow}

\usepackage{titlesec}
\usepackage{lipsum}

\newtheorem{theorem}{Theorem}[section]

\newtheorem{proposition}[theorem]{Proposition}

\newcommand{\ccode}[2]{\par
        \vspace*{8pt}
        {{\leftskip18pt\rightskip\leftskip
        \noindent{\it #1}\/: #2\par}}\par}
\newcommand{\keywords}[1]{\ccode{Keywords}{#1}}
\newcommand{\email}[1]{\href{mailto:#1}{#1}}

\title{\textcolor{Navy}{\textsc{A fractional model for the COVID-19 pandemic: Application to Italian data}}}

\author[1,2]{Elisa Al\`{o}s\thanks{Corresponding author, \email{elisa.alos@upf.edu}}}
\author[3]{Maria Elvira Mancino}
\author[4,5]{Ra\'{u}l Merino}
\author[6]{Simona Sanfelici}

\affil[1]{Dpt. Economia i Empresa, Universitat Pompeu Fabra, \authorcr c/Ramón Trias Fargas, 25-27, Barcelona, Spain.}
\affil[2]{Barcelona Graduate School of Economics, \authorcr c/Ramón Trias Fargas, 25-27, Barcelona, Spain.}
\affil[3]{Department of Economics and Management, University of Florence, \authorcr Via delle Pandette 32,50127 Florence, Italy.}
\affil[4]{Facultat de Matem\`{a}tiques i Inform\`{a}tica, Universitat de Barcelona, \authorcr Gran Via 585, 08007 Barcelona, Spain.}
\affil[5]{VidaCaixa S.A., Market Risk Management Unit, \authorcr C/Juan Gris, 2-8, 08014 Barcelona, Spain.}
\affil[6]{Department of Economics and Management, University of Parma, \authorcr Via J.F. Kennedy, Parma - Italy.}

\date{\normalfont\small\today}

\begin{document}

\maketitle

\begin{abstract}
We provide a probabilistic SIRD model for the COVID-19 pandemic in Italy, where we allow the infection, recovery and death rates to be random. In particular, the underlying random factor is driven by a fractional Brownian motion. Our model is simple and needs only some few parameters to be calibrated.
\end{abstract}

\ccode{MSC classification}{92 C 60, 92 D 30,  60 G 22}
\ccode{JEL classification}{C 02, C 32, C 63, I12}


\clearpage

\allowdisplaybreaks

\section{Introduction}
Since the start of the COVID-19 pandemic many models have been proposed in the literature to explain its  dynamics. Most of these approaches are compartmental models, where the population is divided into compartments that describe the different situations regarding the infectious disease (like susceptible,  infected, etc.), and people progress between them. The basic reference compartmental model is the SIR model, where the population is divided into susceptible (S), infected (I) and recovered (R).
As this model is too simple to describe the complexity of several epidemics, some extensions have been proposed in the literature. For example , the SIRD model considers also the compartment of deaths (D), and the SEIR introduces exposed (E). Some other classical models, like the SEIRS, also allow to model the lost of immunity after recovery. Some recent approaches pay special attention to non directly observable compartments as asymptomatic (A) (see, for example, \cite{IFRV} and the references therein) that, even not being directly observable, play a crucial role in the pandemic. Other extensions consider time-dependent parameters, as in \cite{C}, where the kinetic of the rates of infection ($\beta$) and death ($\mu$) are modeled by exponential functions, while the recovery rate ($\gamma$) is of logistic type. Recent literature also exploits stochastic models, where the main variables account for a noise (Brownian) component, or branching processes \cite{Changguo, Plank, Platen, Yanev}.

\vspace{0.5cm}

In this paper, we introduce a probabilistic SIRD model, where we allow the coefficients to be stochastic  processes determined by a few number of parameters. This randomness is a way to accommodate the effect of several unobservable factors, as the loose of immunity or the existence of asymptomatic. Our construction of the model is motivated by the descriptive analysis of the parameter time series, where we observe that the $\beta$ returns have negatively correlated increments, an observation that suggests to model them  by a fractional Brownian motion (fBm) with a Hurst parameter $H<\frac12$ \cite{Mandelbrot}. We recall that this approach does not need to consider a drift for $\beta$, but this drift arises simply by the properties of the fBm. Once modeled $\beta$, we see that simple relationships between the evolution of diagnosed and infected, allow us to model $\gamma$ and $\mu$.

\vspace{0.5cm}

The model is calibrated in such a way the mean paths of infected, death and recovered fit the corresponding observed values, as well as the variability of $\beta,\gamma$, and $\mu$. Our approach is not only able to adjust observed data, but also to study the different possible scenarios according to its stochastic behavior.

\vspace{0.5cm}

The paper is organized as follows. In Section 2 we present a descriptive analysis of $\beta$, $\gamma$ and $\mu$ corresponding to the evolution of the pandemic in Italy ( \url{https://raw.githubusercontent.com/pcm-dpc/COVID-19/master/dati-andamento-nazionale/dpc-covid19-ita-andamento-nazionale.csv}). Section 3 is devoted to present our stochastic SIRD model for the Italian COVID-19 outbreak. The model is calibrated in Section 4, while we simulate some different scenarios in Section 5. Finally, a discussion of the results and proposals of future research are presented in Section 6.

\section{A descriptive analysis of the SIRD model in the Italian COVID-19 outbreak}\label{Section descriptive}
Let us consider a stochastic SIRD model of the form
\begin{equation}
\begin{cases}
S_{n+1}=S_n-\beta_n I_n S_n/N\\
I_{n+1}=I_n(1+\beta_n S_n/N -\gamma_n-\mu_n)\\
R_{n+1}=R_n+\gamma_n I_n\\
D_{n+1}=D_n+\mu_n I_n,
\end{cases}
\end{equation}
where $S,I,R,D=\{S_n,I_n,R_n,D_n, n=1,...,150\}$ denote the number of daily observed susceptible, infected, recovered and death, N is the population size and $\beta, \gamma,\mu=\{\beta_n,\gamma_n,\mu_n, n=1,...,150\}$ represent the rates of infection, recovery and death. As $S/N\approx 1$ in all the data set, we consider the following simpler version of the above model
\begin{equation}
\label{stochasticSIRD}
\begin{cases}
S_{n+1}=S_n-\beta_n I_n \\
I_{n+1}=I_n(1+\beta_n -\gamma_n-\mu_n)\\
R_{n+1}=R_n+\gamma I_n\\
D_{n+1}=D_n+\mu I_n,
\end{cases}
\end{equation}
Now we observe the behavior of this model in Italy in the period from the 24/2/2020 to the 28/7/2020. In Figure \ref{Fig:Infected-Death-Recovered} we can observe the paths of $I$ ({\it  totale positivi}  in the data set), $R$ ({\it dimessi-guariti}), $D$ ({\it deceduti}) and the total number of diagnosed ($I+R+D$, {\it totale casi}).
\begin{figure}[H]
\centering
\includegraphics[scale=0.18]{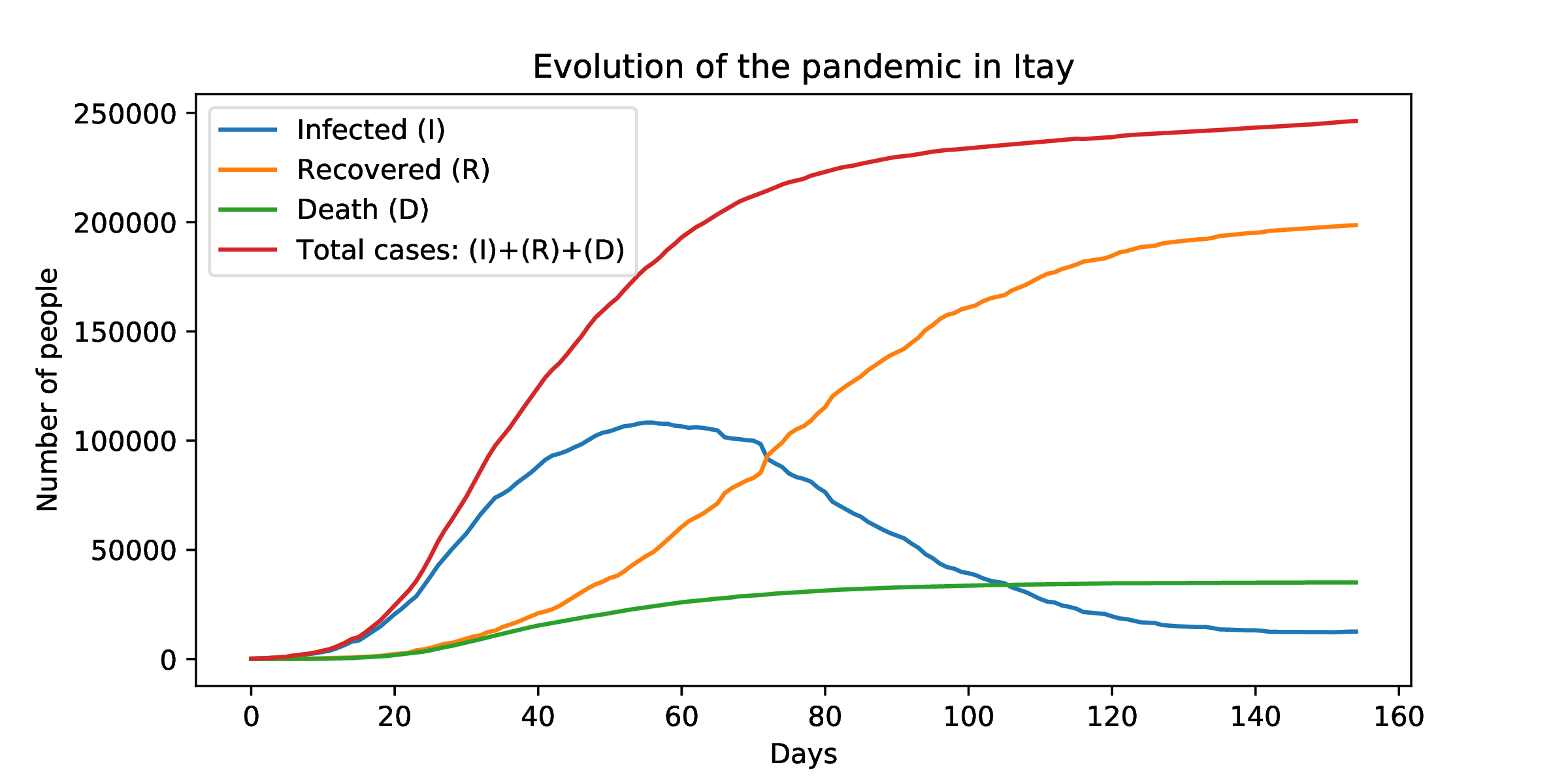}
\caption{Evolution pandemic in Italy.}
\label{Fig:Infected-Death-Recovered}
\end{figure}

\subsection{The process \texorpdfstring{$\beta$}{beta} }
Now let us observe $\beta, \gamma,\mu$. In Figure \ref{Fig:beta} we can see the behavior of $\beta$, that is a decreasing function of time. This fits what expected due to the lock-down, that reduced the number of contacts between individuals. In Figure \ref{Fig:Increments beta}, we can see the corresponding increments (i.e., $\Delta \beta= \beta_{n+1}-\beta_n$), that are more variable at the beginning of the sample, when $\beta$ is higher.
\begin{figure}[H]
\centering
\includegraphics[scale=0.18]{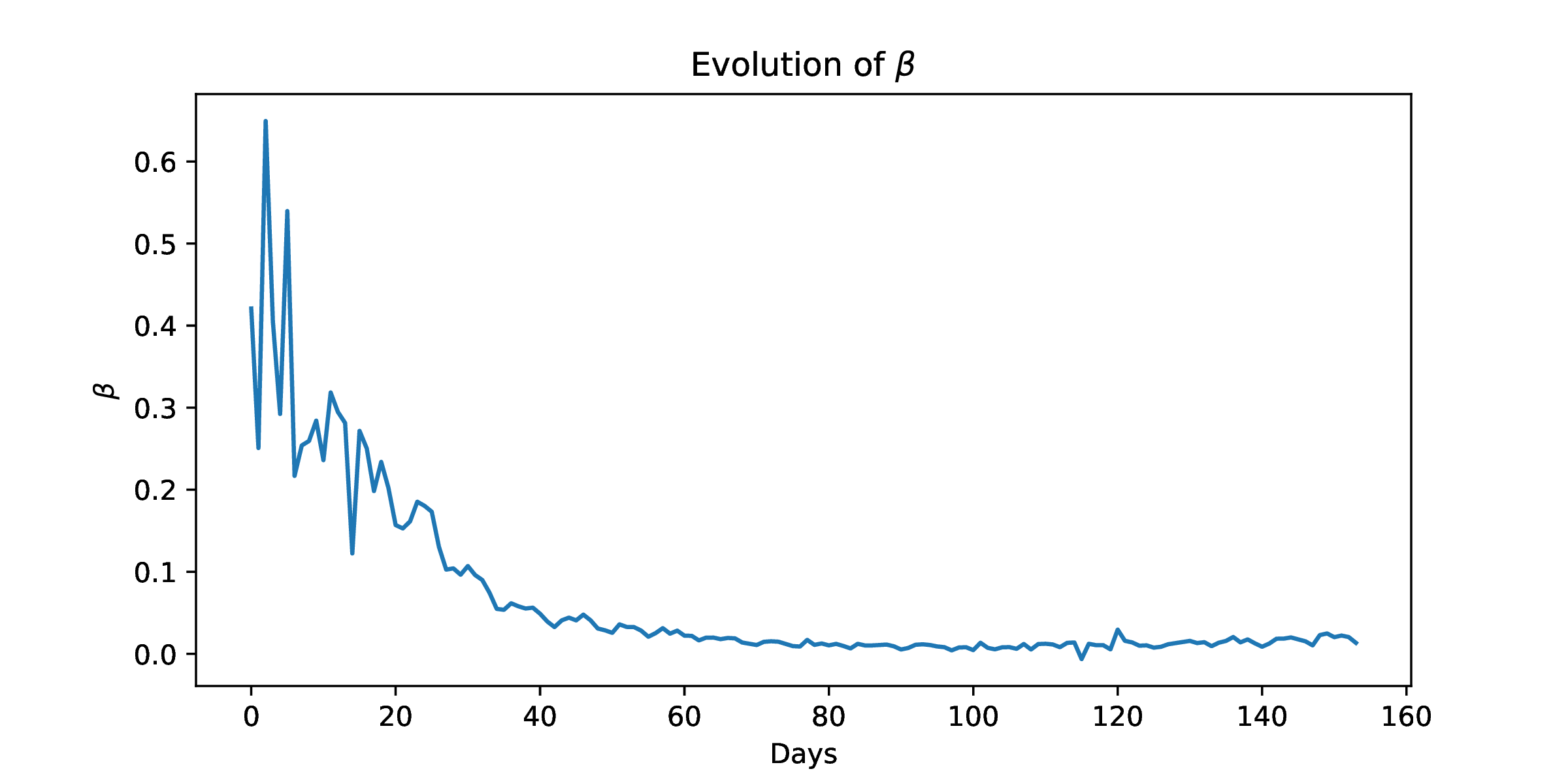}
\caption{Evolution of $\beta$.}
\label{Fig:beta}
\end{figure}
\begin{figure}[H]
\centering
\includegraphics[scale=0.18]{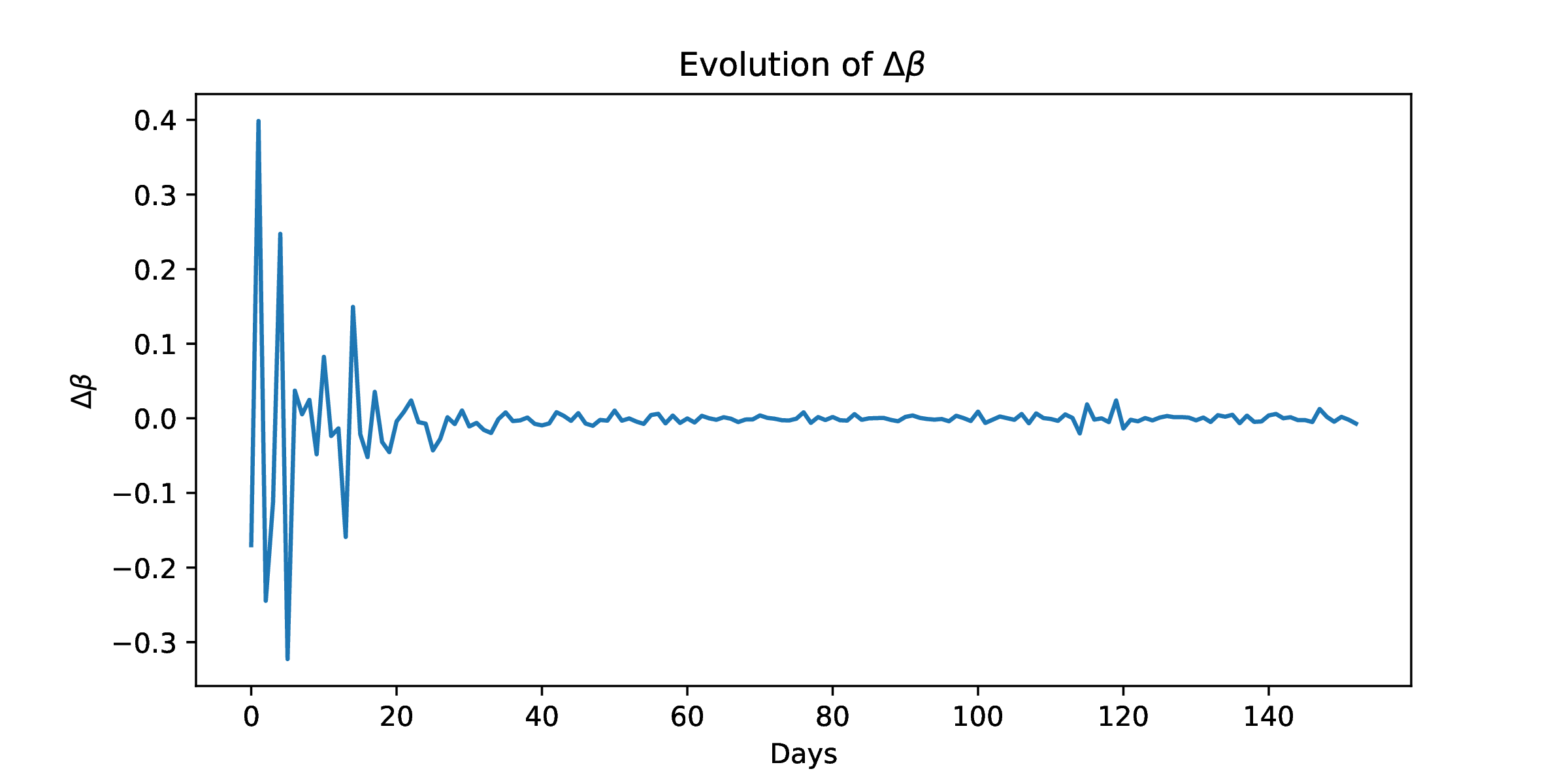}
\caption{Increments of $\beta$.}
\label{Fig:Increments beta}
\end{figure}

Moreover, in Figure \ref{Fig:Evolution of the beta returns} we can see the returns $\frac{\Delta \beta}{\beta}$ that have a stationary mean. After deleting the outliers at days 115 and 119 (making them equal to zero), the corresponding autocorrelation function, see Figure \ref{Fig:beta autocorrelation function}, reveals a negative correlation between increments. The empirical analysis suggests to model $\beta$ as a process with negative correlated increments with a decreasing mean. This process is designed in Section 3.
\begin{figure}[H]
\centering
\includegraphics[scale=0.18]{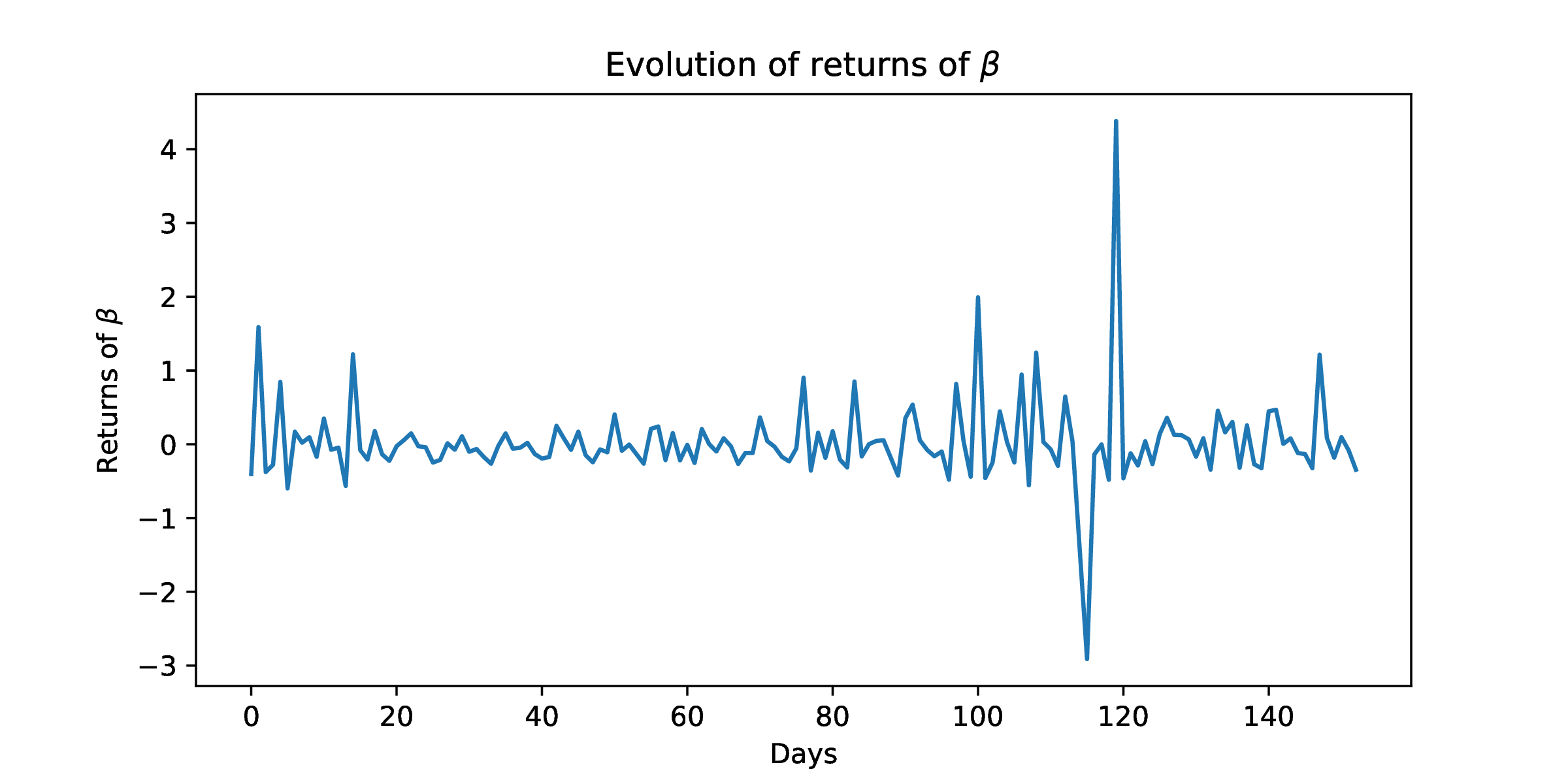}
\caption{Evolution of the returns of $\beta$.}
\label{Fig:Evolution of the beta returns}
\end{figure}
\begin{figure}[H]
\centering
\includegraphics[scale=0.18]{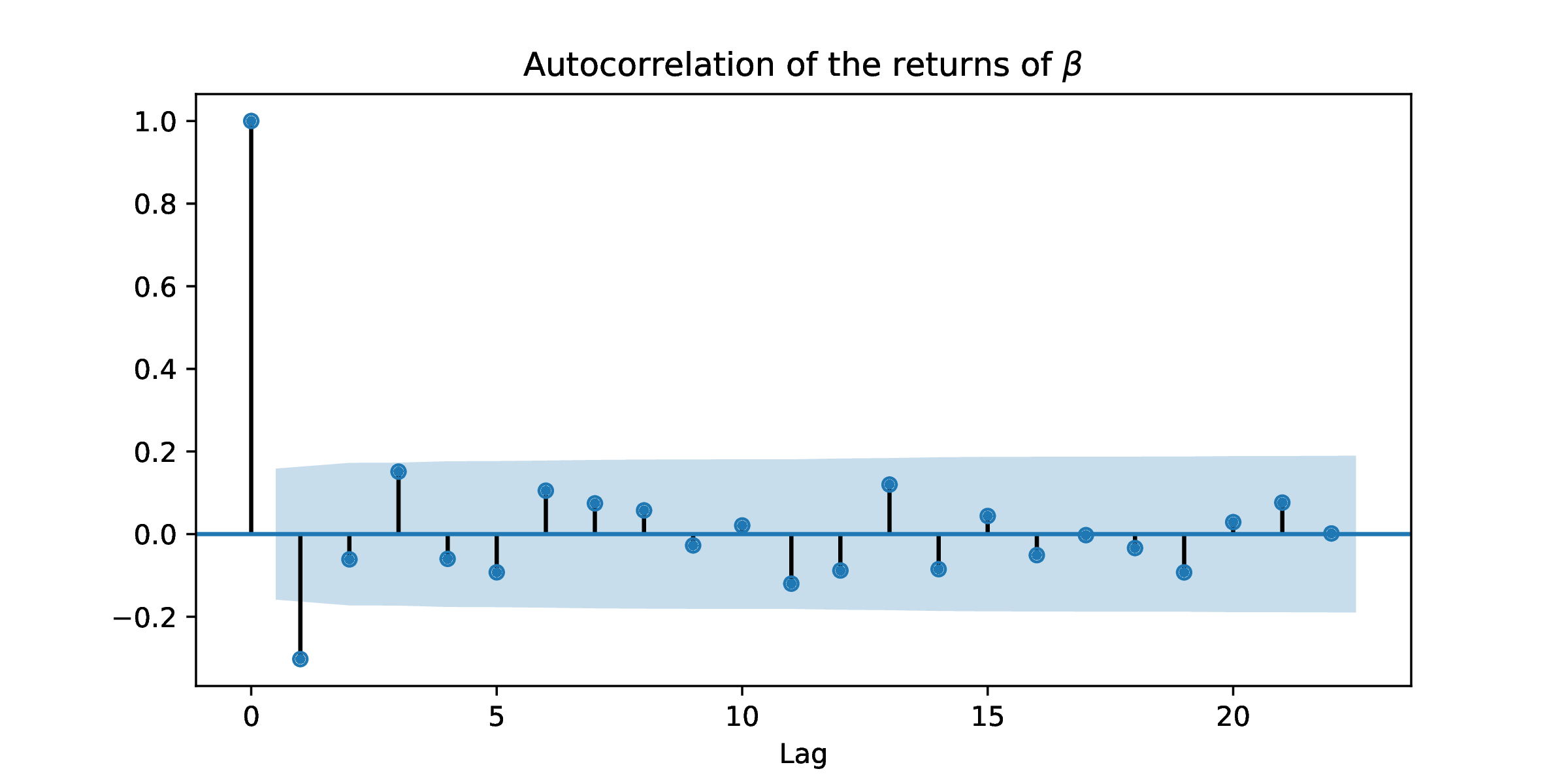}
\caption{Autocorrelation function of $\beta$ returns.}
\label{Fig:beta autocorrelation function}
\end{figure}

\subsection{The process \texorpdfstring{$\mu$}{mu}}
The time behavior of $\mu$ can be seen in Figure \ref{Fig:Evolution of mu}. We observe that it is similar to the behavior of $\beta$, but with a delay. Moreover, Figures \ref{Fig:infected vs death} and \ref{Fig:infected vs death2} show a clear relationship between the daily new diagnosed and deaths, with a delay of 4 days. \footnote{This finding appears consistent with data by Istituto Superiore di Sanit\`{a} \url{https://www.epicentro.iss.it/coronavirus/bollettino/Report-COVID-2019_17_marzo-v2.pdf}.}
In Figure \ref{Fig:diag vs death1}, we can see a strong link between total diagnosed and deaths. In  Figures \ref{Fig: Linear Dep delta infected vs death} and
\ref{Fig: Linear Dep infected vs death} we can see that this relationship is linear.
Then, as the number of new infected is given by $\beta I$, the daily increments of deaths ($\Delta D$) would be modeled simply as  $c_\mu\beta_{n-4}I_{n-4}$, for some positive constant $c_\mu$. This leads to the following model for $\mu$:
$$\mu_n=c_\mu\beta_{n-4}I_{n-4}/I_n,$$

\begin{figure}[H]
\centering
\includegraphics[scale=0.18]{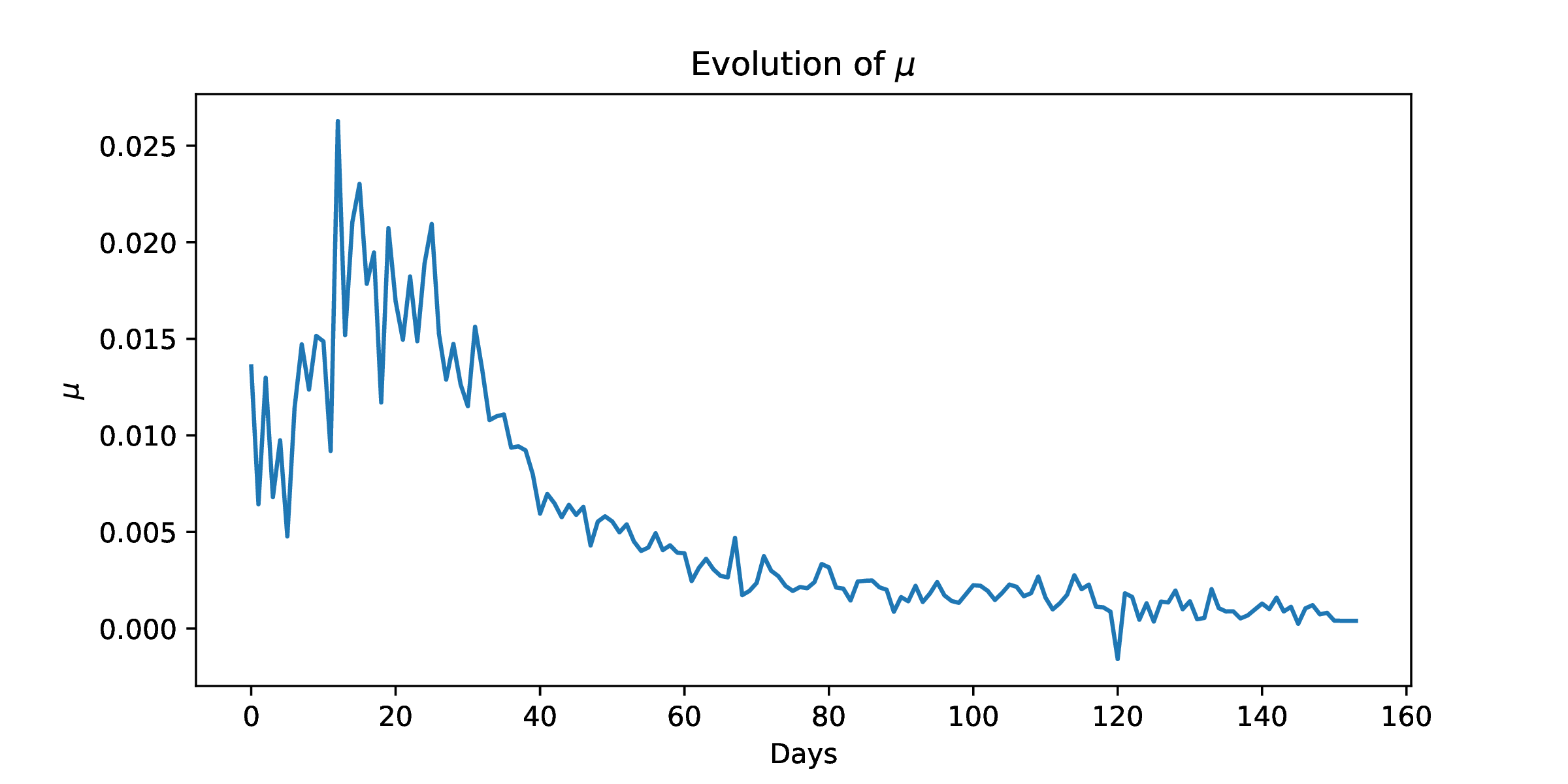}
\caption{Evolution of $\mu$.}
\label{Fig:Evolution of mu}
\end{figure}

\begin{figure}[H]
\centering
\includegraphics[scale=0.18]{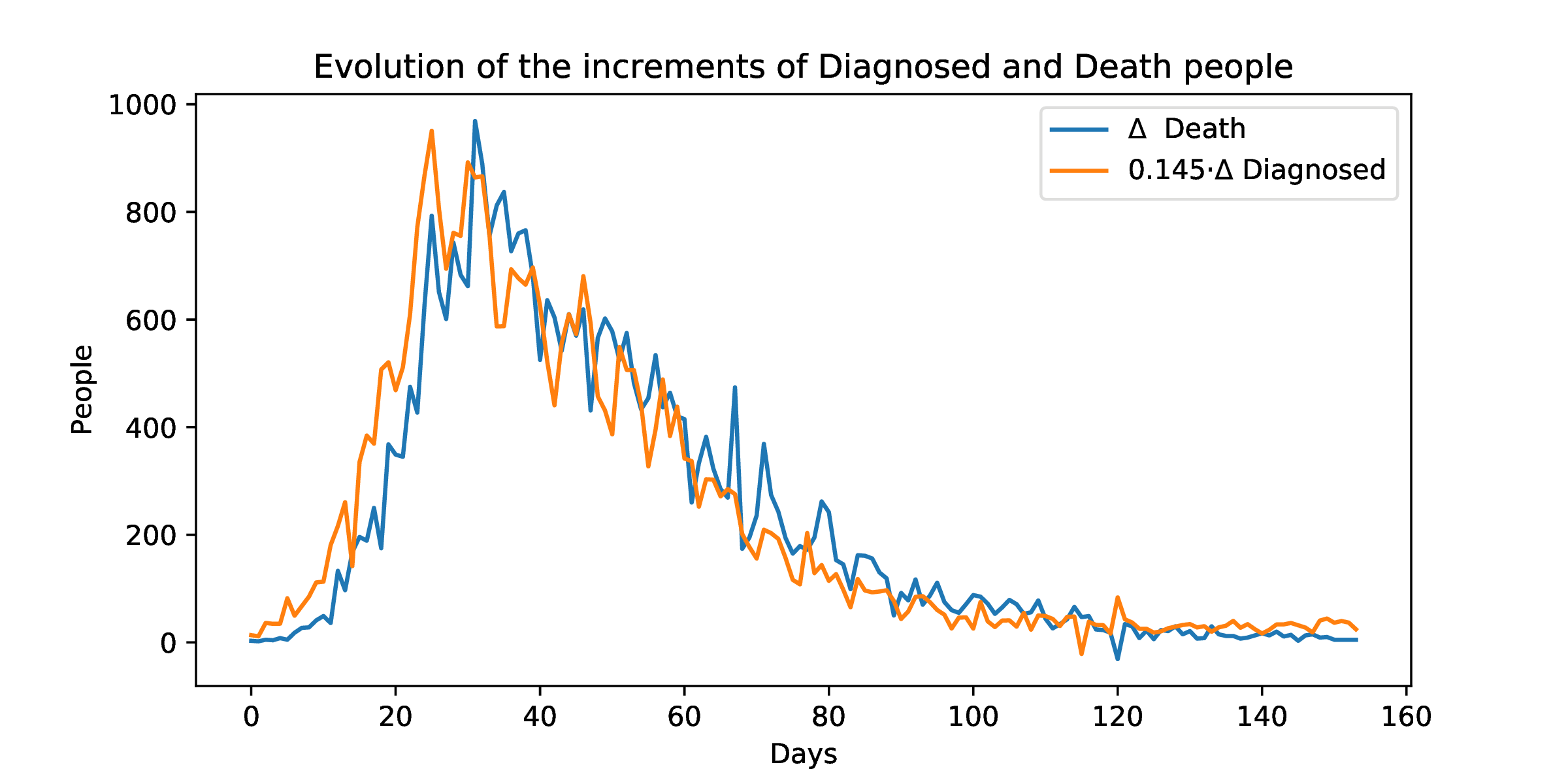}
\caption{Comparing the time series of 0.14$\cdot$ $\Delta$ Infected against $\Delta$ Death.}
\label{Fig:infected vs death}
\end{figure}

\begin{figure}[H]
\centering
\includegraphics[scale=0.18]{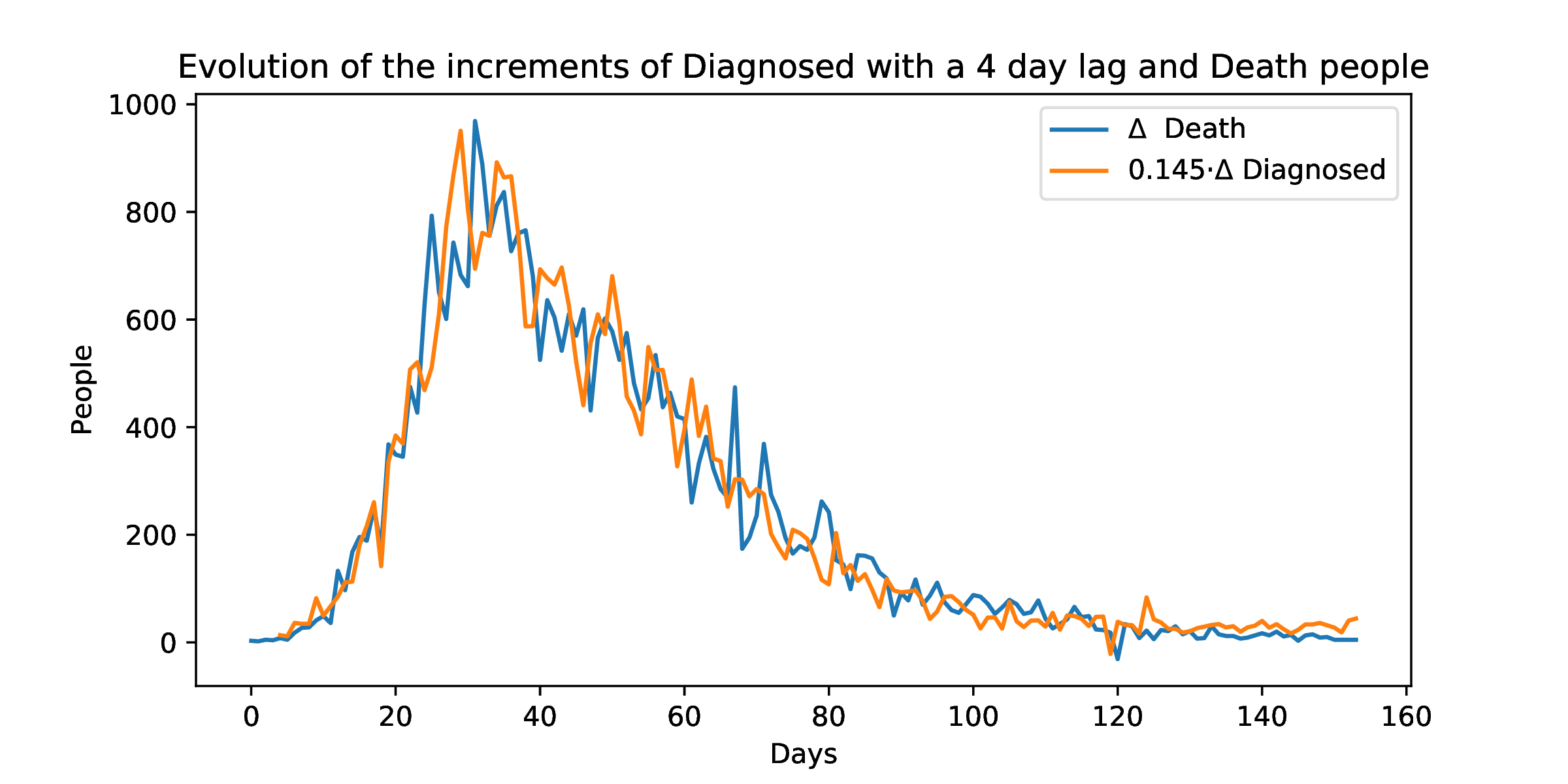}
\caption{Comparing the time series of 0.14$\cdot$ $\Delta$ Infected against $\Delta$ Death.}
\label{Fig:infected vs death2}
\end{figure}

\begin{figure}[H]
\centering
\includegraphics[scale=0.18]{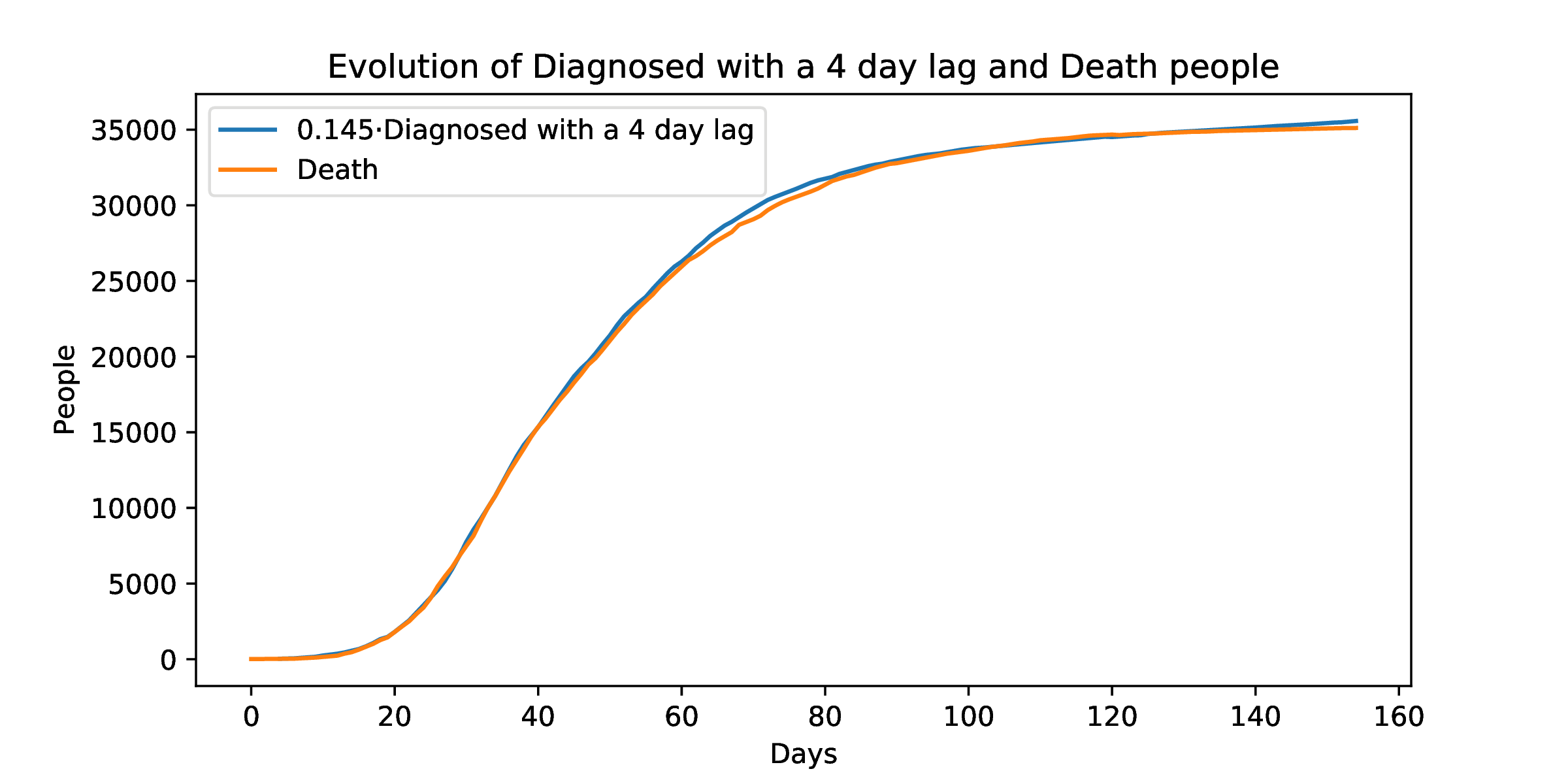}
\caption{Comparing the time series of diagnosed (I+R+D) with 4 days lag against Death.}
\label{Fig:diag vs death1}
\end{figure}

\begin{figure}[H]
\centering
\includegraphics[scale=0.18]{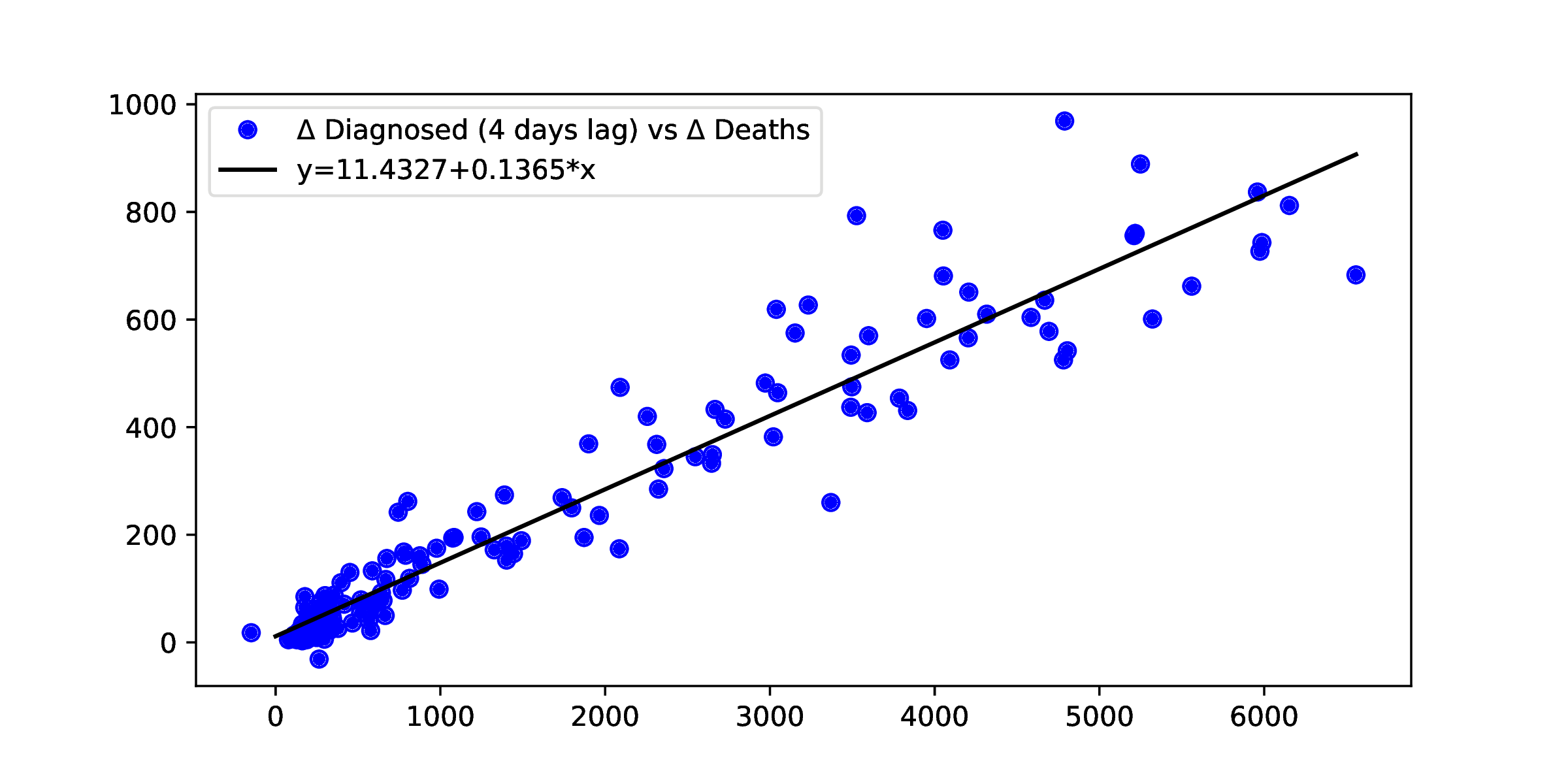}
\caption{Linear Dependence between $\Delta$  Diagnosed (with 4 days lag) and $\Delta$ Death}
\label{Fig: Linear Dep delta infected vs death}
\end{figure}

\begin{figure}[H]
\centering
\includegraphics[scale=0.18]{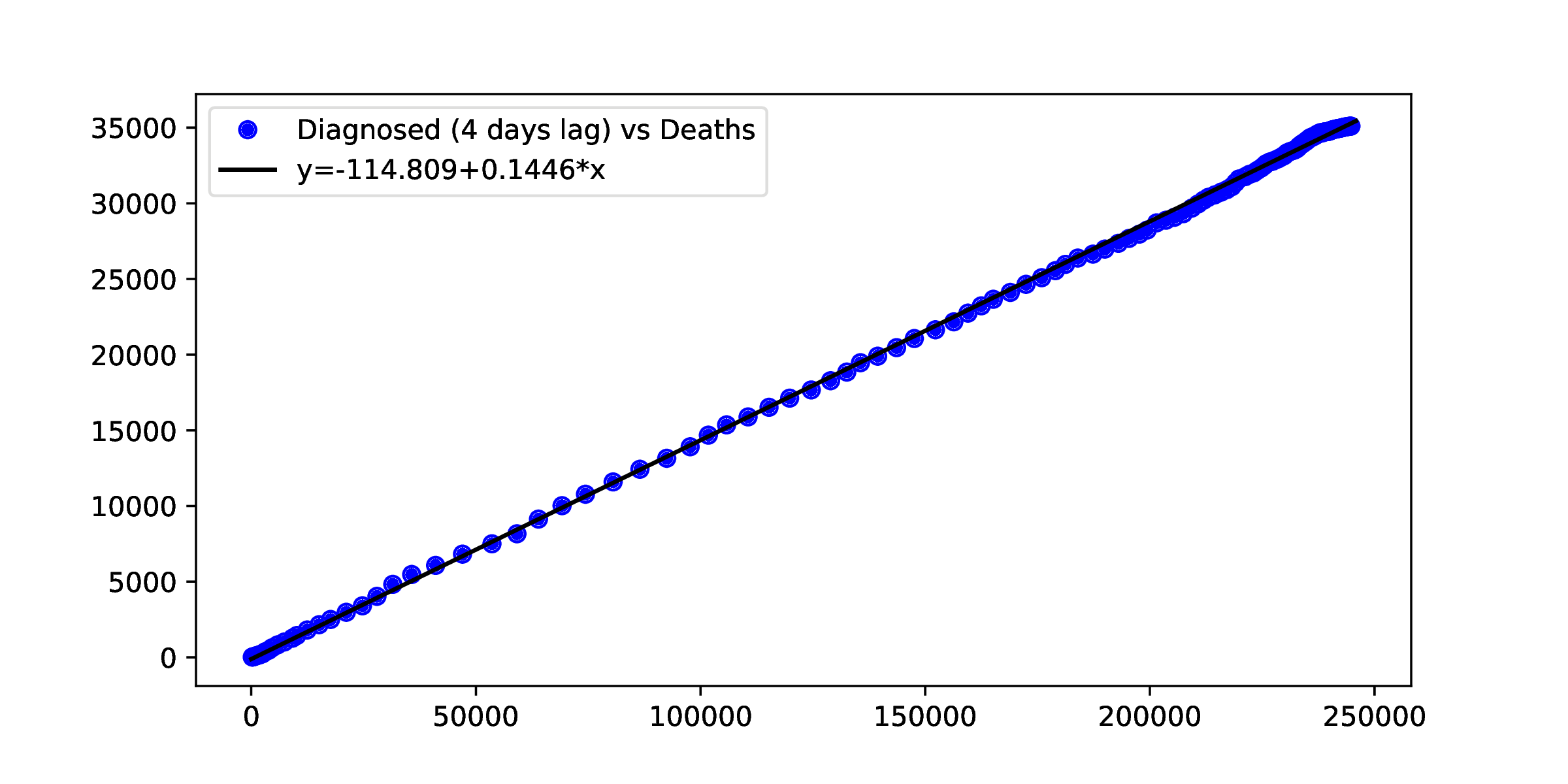}
\caption{Linear Dependence between total number of Diagnosed (with 4 days lag) and Deaths}
\label{Fig: Linear Dep infected vs death}
\end{figure}

\subsection{The process \texorpdfstring{$\gamma$}{gamma}}

Finally, let us observe the path of $\gamma$ in  Figure \ref{Fig:Evolution of gamma}. This recovery rate seems to move between an upper and a lower bound, that we can interpret as the  rate under stressed conditions of the health system, and the rate when this system adapts to the new scenario and the epidemic curve decays. This is connected to the approach in \cite{C}, where $\gamma$ is assumed to be described by a logistic function, and the rate of recovery increases from a initial value $\gamma_0$ and it stabilizes at some higher level.

\begin{figure}[H]
\centering
\includegraphics[scale=0.18]{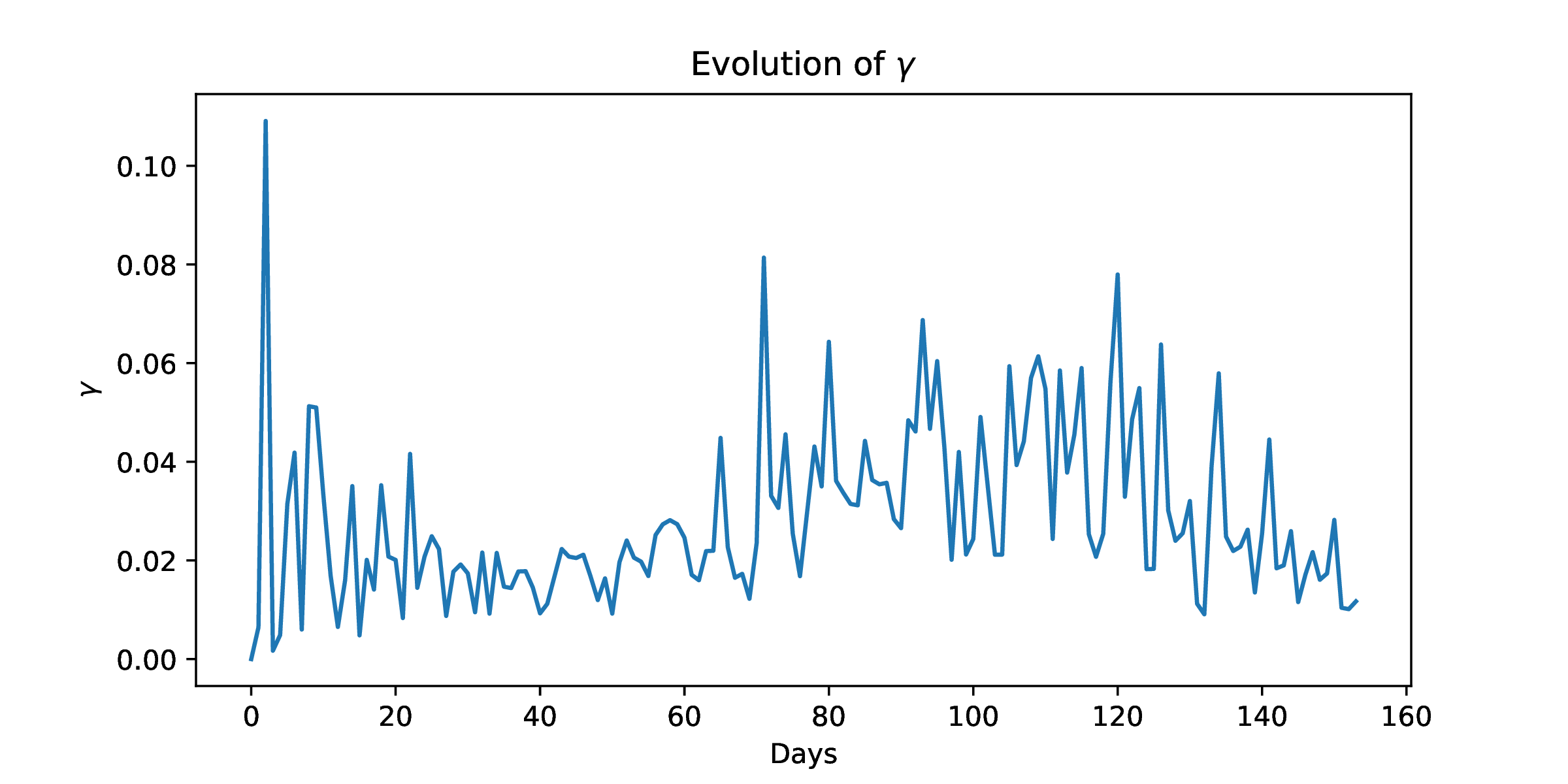}
\caption{Evolution of $\gamma$.}
\label{Fig:Evolution of gamma}
\end{figure}

A measure of the stress of the model can be defined as the quantity
\begin{equation}
\label{stress}
(I+R+D)_{-30n}/(I+R+D)_n,
\end{equation}
 where $(I+R+D)_{-30n}$ is the average of the diagnosed people the last $30$ days. If there are no new infections, the number of diagnosed remains stable, and then the above quantity is near 1. Otherwise, a high increment of new cases translates into a decrease of the value of this ratio. In Figure \ref{logistic}, we can see the behavior of this process
\begin{figure}[H]
\centering
\includegraphics[scale=0.18]{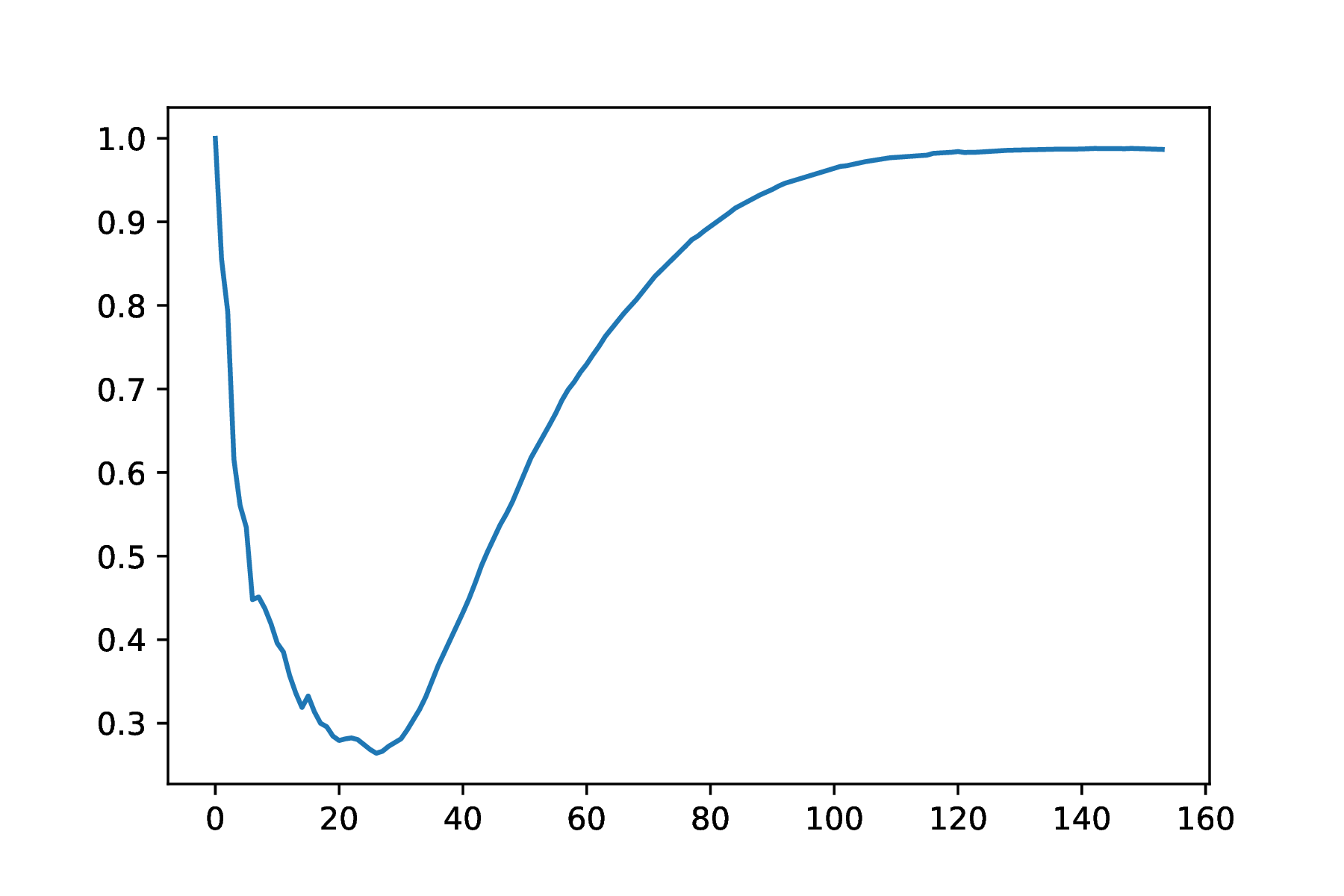}
\caption{$(I+R+D)_{-30n}/(I+R+D)_n$}
\label{logistic}
\end{figure}

We can also observe in Figure \ref{Fig:Evolution of gamma} a big variance, as expected due to the fact that  not everybody recovers at the same speed. Then we would like to add some noise, multiplying this time series by an adequate random factor, as we detail in the following section.

\section{A stochastic SIRD model}
We propose a stochastic SIRD model where randomness is embedded by modeling the parameters of infection ($\beta$), mortality ($\mu$) and recovery ($\gamma$). In the following subsections, we will describe the model for each one of the parameters.

\subsection{A stochastic model for \texorpdfstring{$\beta$}{beta}.}
As the return of $\beta$  are negatively correlated, it is natural to construct a model based on the fractional Brownian motion (fBm). The fBm  is a Gaussian process with
stationary increments, which depends on a parameter $H\in (0,1)$ called the Hurst parameter.
More precisely, a process $B^H={B_t^H, t\in [0,T]}$ is called a fractional Brownian motion (fBm) if
\begin{itemize}
\item $E(B_0^H)=0.$
\item $E(B_t^H B_s^H)=\frac12(s^{2H}+t^{2H}-|t-s|^{2H})$.
\end{itemize}
Except in the Brownian motion case $H=1/2$, the fBm increments are correlated. Denote $X_n=B_n^H-B_{n-1}^H$
and define $\rho_H(n):=Cov (X_1,X_{n+1})$. Then we have

\begin{eqnarray}
\label{correlation}
\rho_H(n)&=&E(B_1^H(B_{n+1}^H-B_n^H))\nonumber\\
&=&\frac12(1^{2H}+(n+1)^{2H}-n^{2H})\nonumber\\
&&-\frac12(1^{2H}+n^{2H} -(n-1)^{2H})\nonumber\\
&&=\frac12\left((n+1)^{2H}+(n-1)^{2H}-2n^{2H}\right)
\end{eqnarray}
Notice that this quantity is positive for $H>\frac12$ and negative if $H<\frac12$.

\medbreak

Then, a natural candidate to model $\beta$ is given by
\begin{equation}
\label{eqbeta}
\hat{\beta}_{n+1}=\hat{\beta}_{n}(1+(B_{n+1}^H-B_n^H)),
\end{equation}
where $B^H$ is a fractional Brownian motion with $H<\frac12$ and $k$ is a constant. Apart from the negative correlated returns, this process has also a decreasing mean, as we prove in the following result.
\begin{proposition}
Consider a fractional Brownian motion $B^H$ with $H<\frac12$ and the process defined by (\ref{eqbeta}). Then, provided $\beta_0>0$,
$E(\hat{\beta}_n)$ is decreasing as $n\to\infty$.
\end{proposition}
\begin{proof}
A recursive computation gives us that
$$
\hat{\beta}_n:=\hat{\beta}_0 A_n,
$$
with
$$
A_n:=\Pi_{m=1}^n (1+k\Delta B_m^H),
$$
where $\Delta B_m^H=B_{m+1}^H-B_m^H$.
The above product can be written as
\begin{equation}
1+\sum_{i=1}^n k^i\left(\sum_{m_1<\cdot\cdot\cdot <m_i} \Delta B_{m_1}^H \cdot\cdot\cdot\Delta B_{m_i}^H \right)
\end{equation}
Now, we can compute the expectation using Isserlis Theorem. This gives us that
\begin{eqnarray}
&&E(A_m)\nonumber\\
&&=E\left(\Pi_{m=1}^n (1+k\Delta B_m^H)\right)\nonumber\\
&&=1+\sum_{i=1}^n k^i\left(\sum_{m_1<\cdot\cdot\cdot< m_i} E(\Delta B_{m_1}^H \cdot\cdot\cdot\Delta B_{m_i}^H )\right)\nonumber\\
&&=1+\sum_{i=1}^n k^i\left(\sum_{m_1<\cdot\cdot\cdot <m_i}\sum_{p\in \mathcal{P}_i}\Pi_{(k,j)\in p}Cov( \Delta B_{m_k}^H , \Delta B_{m_j}^H ) \right)\nonumber\\
&&=1+\sum_{i=2l}^n k^i\left(\sum_{m_1<\cdot\cdot\cdot <m_i}\sum_{p\in \mathcal{P}_i}\Pi_{(k,j)\in p}\rho_H(| m_j-m_k|)\right),\nonumber\\
\end{eqnarray}
where we have used that partitions $\mathcal{P}_i$ exist only if $i$ is even. Then, as $k^i$ is positive for all even $i$ and all the $\rho_H(| m_j-m_k|)$ are strictly negative, $E(A_n)$ is decreasing. This allows us to complete the proof.

\end{proof}
Even when the model (\ref{eqbeta}) reproduces some empirical properties of $\beta$, a numerical analysis shows that it has to be modified before being an adequate model. More precisely, we can see in Figure \ref{Fig:beta_simulation} that the variability of the paths is high, and that the process
$\hat{\beta}$ can even become negative (see also Figure \ref{Fig:Beta max min}). Moreover, the estimated correlation between two consecutive increments of  $\hat{\beta}$ does not coincide with the observed for the $\beta$ returns.

\begin{figure}[H]
\centering
\includegraphics[scale=0.18]{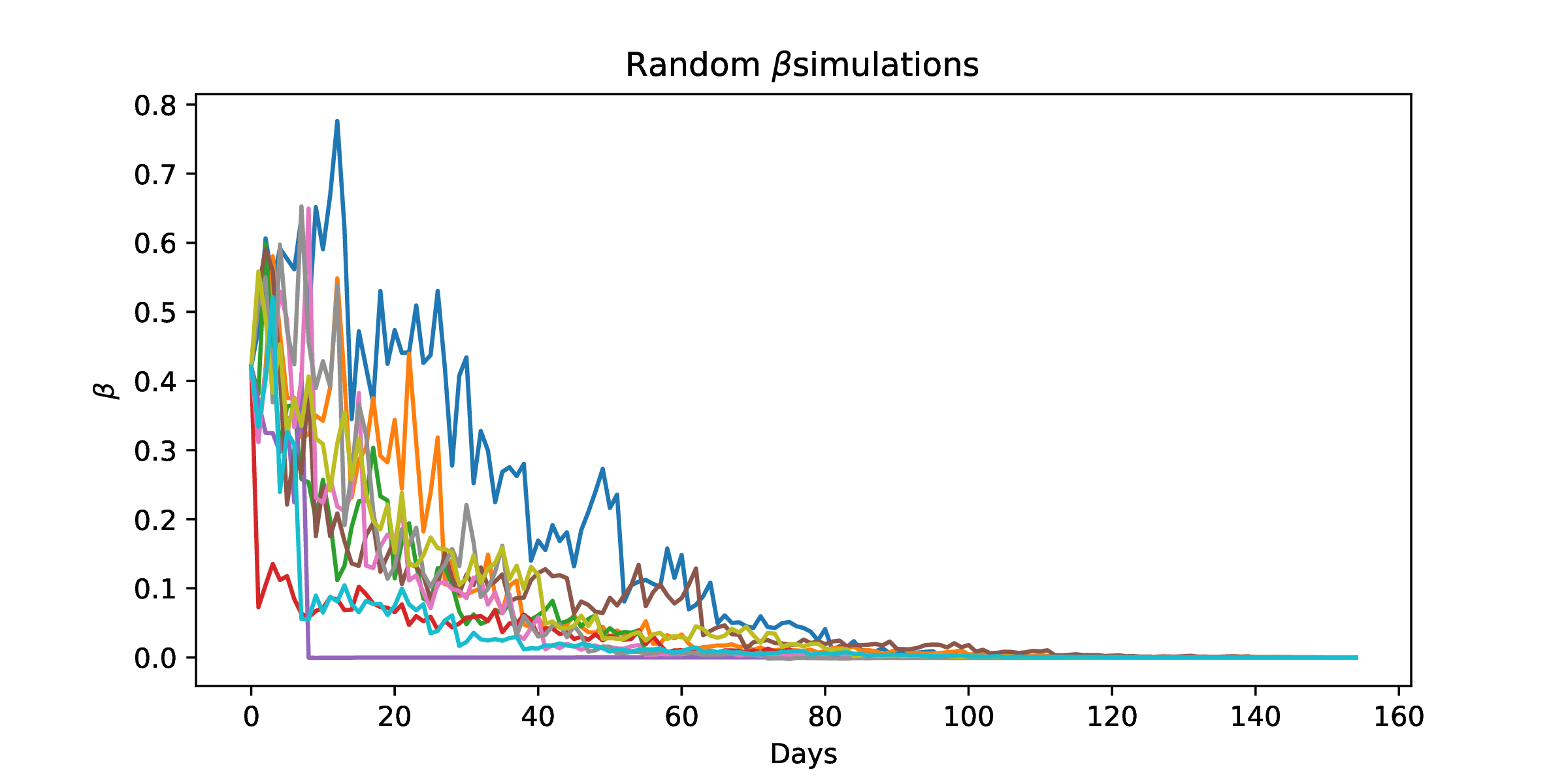}
\caption{10 possible paths of $\hat\beta$ for $H=0.1$.}
\label{Fig:beta_simulation}
\end{figure}

A way to reduce the variance of the paths is obviously to define a model of the type
\begin{equation}
\label{eqbeta}
\beta_{n+1}=\sum_{i=1}^m \beta_{n}^i,
\end{equation}
where $\beta^i, i=1,..,m$ are given by
$$
\beta^i_{n+1}=\beta^i_{n}(1+c_\beta^i(B_{n+1}^{i,H}-B_n^{i,H})),
$$
being $B^{i,H}, i=1,...,m$ $m$ independent fractional Brownian motions with Hurst parameter $H<\frac12$ and where $c_\beta^i, i=1,...,n$ are positive constant. The interpretation of (\ref{eqbeta}) is intuitive: the observed process $\beta$ is really the sum of different stochastic $\beta^i$ that correspond to the particular transmission rates in different groups, locations, etc.
 In Figures \ref{Fig:Beta max min} and \ref{Fig:Beta max min m10}we can observe, for $H=0.1$, the behavior of the cases $m=1$ (where the volatility is too big), and $m=10$, where the area delimited by the the maximum and minimum paths contain the observed values of $\beta$.

\begin{figure}[H]
\centering
\includegraphics[scale=0.18]{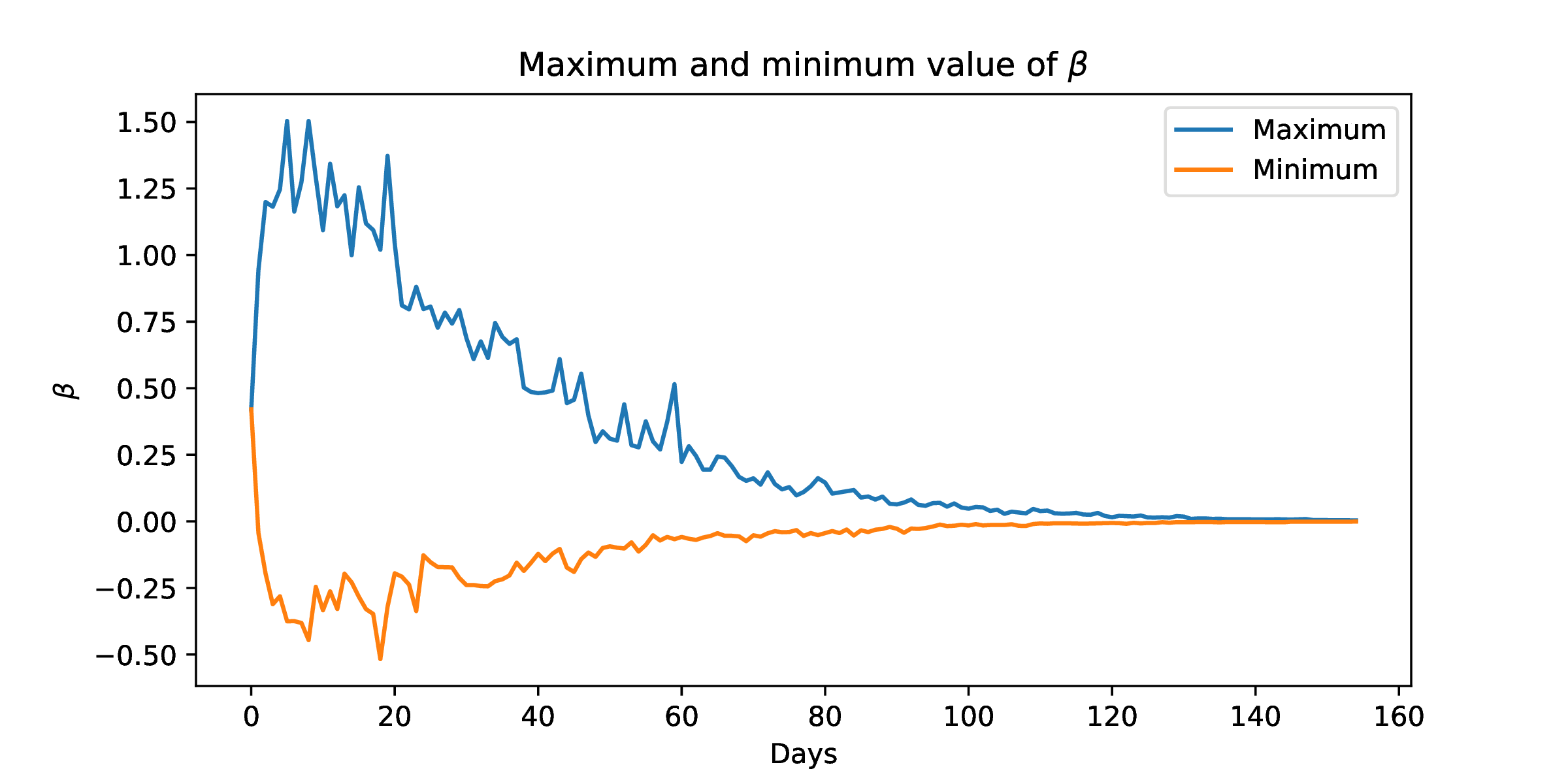}
\caption{Maximum and minimum of the simulated paths of $\hat\beta$.}
\label{Fig:Beta max min}
\end{figure}

\begin{figure}[H]
\centering
\includegraphics[scale=0.18]{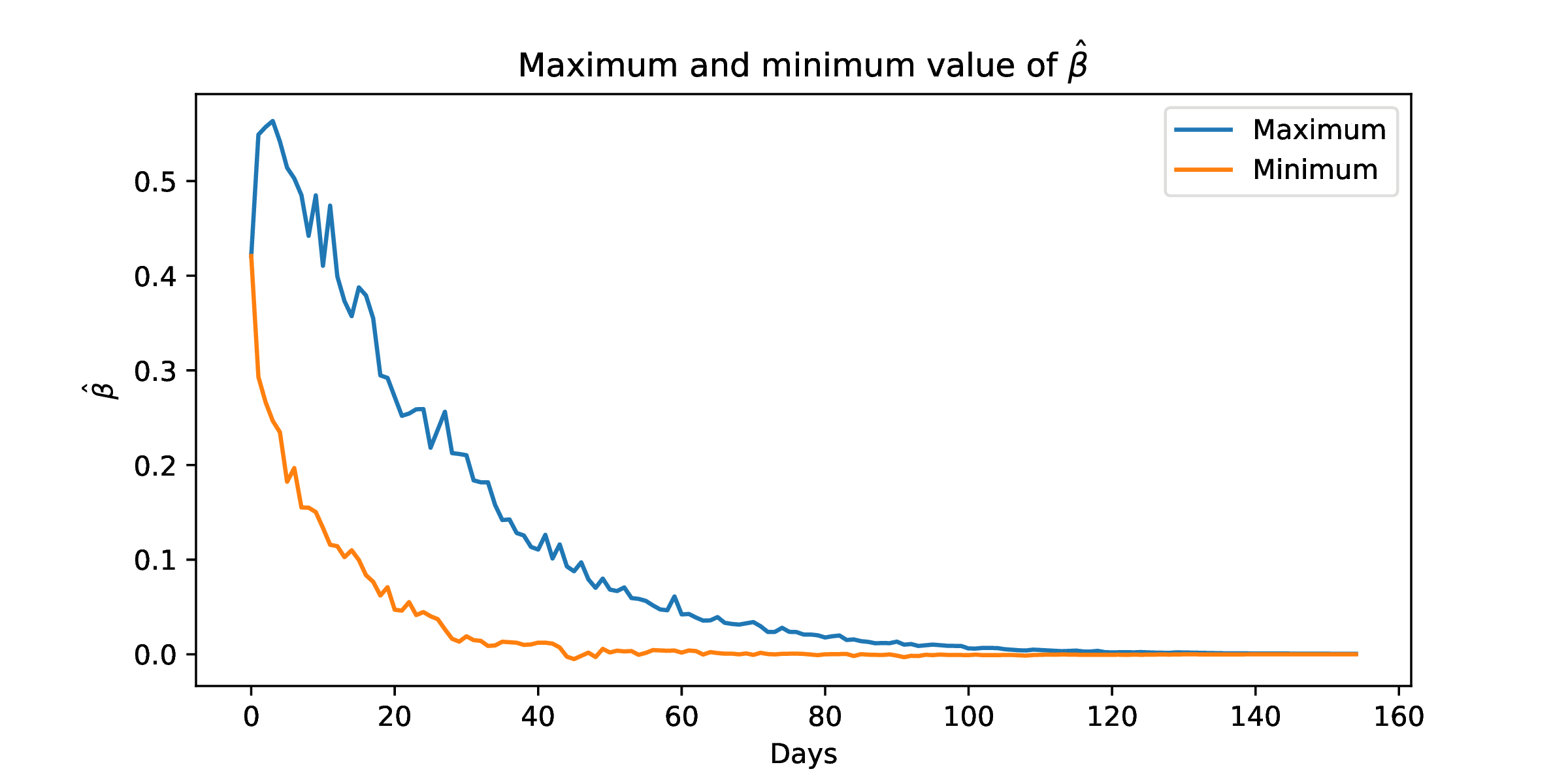}
\caption{Maximum and minimum of the simulated paths of $\beta$ and $m=10$.}
\label{Fig:Beta max min m10}
\end{figure}

\subsection{A stochastic model for \texorpdfstring{$\mu$}{mu} and \texorpdfstring{$\gamma$}{gamma}.}
As pointed out in \ref{Section descriptive}, the linear relationship between diagnosed (with 4 days lag) and death, suggests modeling $\mu$ simply as
$$\mu_n=c_\mu\beta_{n-4}I_{n-4}/I_n.$$ On the other hand, the observations in Section 2 lead to a model for $\gamma$ of the type
$$
\gamma_n=c_\gamma^1[(I+R+D)_{-30n}/(I+R+D)_n]\exp(c_\gamma^2 (B_{n+1}^{H_{\gamma}}-B_n^{H_{\gamma}})),
$$
for some positive constants $c_\gamma^1,c_\gamma^2$ and for some adequate Hurst parameter $H_{\gamma}$.

\subsection{The global model}
The above models for $\beta, \gamma$, and $\mu$ lead to the following stochastic SIRD model:
\begin{equation}
\begin{cases}
\beta_{n+1}=\sum_{i=1}^{10} \beta_{n}^i, \hspace{0.2cm}, with \hspace{0.2cm} \beta^i_{n+1}=\beta^i_{n}(1+c_\beta(B_{n+1}^{i,H}-B_n^{i,H}))\\
\mu_{n}=c_\mu\beta_{n-4}\frac{I_{n-4}}{I_n}\\
\gamma_n=c_\gamma^1[(I+R+D)_{-30n}/(I+R+D)_n]\exp(c_\gamma^2 (B_{n+1}^{H_{\gamma}}-B_n^{H_{\gamma}})\\
S_{n+1}=S_n-\beta_n I_n \\
I_{n+1}=I_n(1+\beta_n -\gamma_n-\mu_n)\\
R_{n+1}=R_n+\gamma I_n\\
D_{n+1}=D_n+\mu I_n.
\end{cases}
\end{equation}
Notice that only 7  parameters have to be calibrated are  $m$, $H$, $H_{\gamma}$ and $c_\beta$, $c_\mu$,$c_\gamma^1$, $c_\gamma^2$. We see in the following section how the set of the first three parameters $m$, $H$, $H_{\gamma}$, that define the driving processes of the model,  is chosen based on empirical observations, while the last group is calibrated by means of a classical least squares method.

\section{Calibration}

\subsection{The steps of the calibration process}
The calibration procedure is as follows.

\begin{itemize}
\item In a first step, we fix reasonable values of $m$ and the Hurst parameters according to observed data, and
\item fixed $m$, $H$, and $H_{\gamma}$, we calibrate  $c_\beta$, $c_\mu$,$c_\gamma^1$, $c_\gamma^2$ by a least squares method.
\end{itemize}
Step 1 focuses on choosing adequate driving random processes for the model. As estimating with precision and robustness this group of parameters is not straightforward (see for example Glotter (2007)), so we simply proceed empirically. More precisely, we have seen in Section 2 that the maximum and the minimum paths in the case $H=0.1$ and $m=10$ envelope the observed $\beta$ time series. Moreover, the observed correlation of beta returns (for lag=1)  is equal to $-0.347$, while (from a 1000 simulations sample), we estimate this quantity to be  $-0.314$ for $H=0.1$ and $m=10$. This leads to choose $m=10, H=0.1$ in our model.

\medbreak

In order to choose $H_\gamma$, we compute the autocorrelation function of
$$
\log\left(\frac{\gamma_n}{[(I+R+D)_{-30n}/(I+R+D)_n]}\right).
$$
At lag=1, this autocorrelation function is equal to $0.244$. Then, taking into account Equation (\ref{correlation}), we get a estimation of $H_\gamma=0.657$. Then we choose $H_{\gamma}=0.6$ for the sake of simplicity.

\subsection{Initial guess for  \texorpdfstring{$c_\beta$}{cbeta}, \texorpdfstring{$c_\mu$}{cmu}, \texorpdfstring{$c_\gamma^1$}{cgamma1}, \texorpdfstring{$c_\gamma^2$}{cgamma2} and global calibration}
In order to be able to start the calibration of the other parameters, we will need to have an initial guess.
\medbreak
 To find the parameter $c_\beta$, we are going to fit the average $\hat{\beta}$ against the realized $\beta$. This gives us a value of 0.32250809.
\medbreak
For $\mu$, as we have seen previously, we can obtain a nice estimation by doing a linear regression between the infected and death people. The slope ( 0.13904755) will be the initial guess for $c_\mu$.

\medbreak
In the case of $\gamma$, in order to obtain the variables $c_\gamma^1$, $c_\gamma^2$, we  minimize the distance between the average $\hat{\gamma}$ and the observed $\gamma$, as well as  the distance between the standard deviation of the returns of $\hat{\gamma}$ and the observed $\gamma$. This gives us (0.03512639,0.5) as initial guess.

\medbreak

Once we have the initial guess, we find the best possible parameters to minimize the distance between the infected, the death and the recovered time series as well as the empirical variance of the $\gamma$ by OLS.

\section{Results}
Following the procedure in the previous section, we get the following estimation of the parameters of the model:
$$c_\beta=0.32102563, c_\mu=0.13981687, c_\gamma^1=0.03383898, c_\gamma^2=0.54395798.$$

\subsection{The average paths of the global calibration}
Simulating the model and taking the corresponding mean paths we fit the observed data, as we can see in Figure \ref{global}.

\begin{figure}[H]
\centering
\includegraphics[scale=0.18]{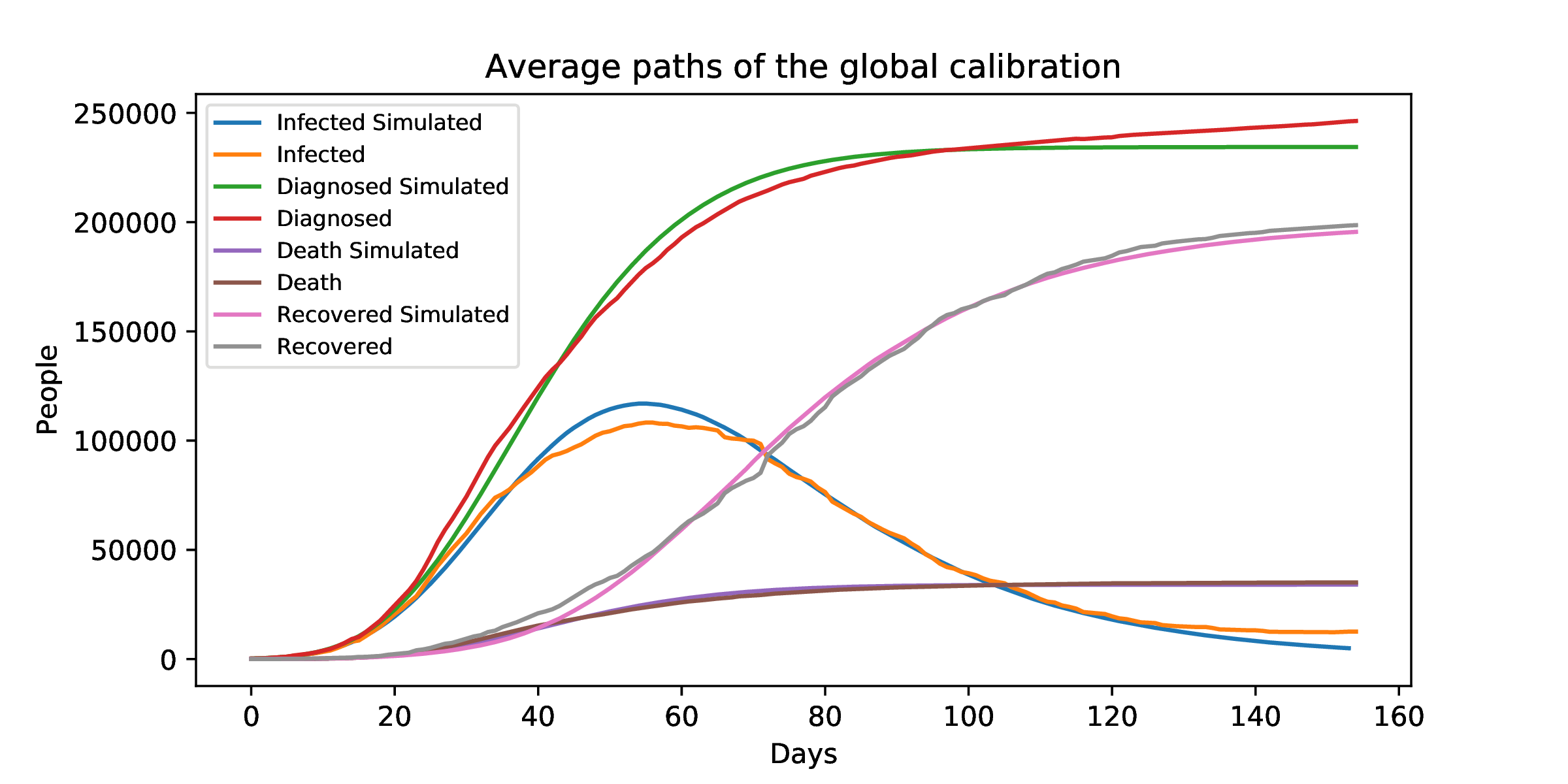}
\caption{Simulated mean paths}
\label{global}
\end{figure}

\subsection{The average paths of \texorpdfstring{$\beta$}{beta}, \texorpdfstring{$\gamma$}{gamma}, and \texorpdfstring{$\mu$}{mu}}
The same analysis can be done for $\beta$, $\gamma$, and $\mu$:

\begin{figure}[H]
\centering
\includegraphics[scale=0.18]{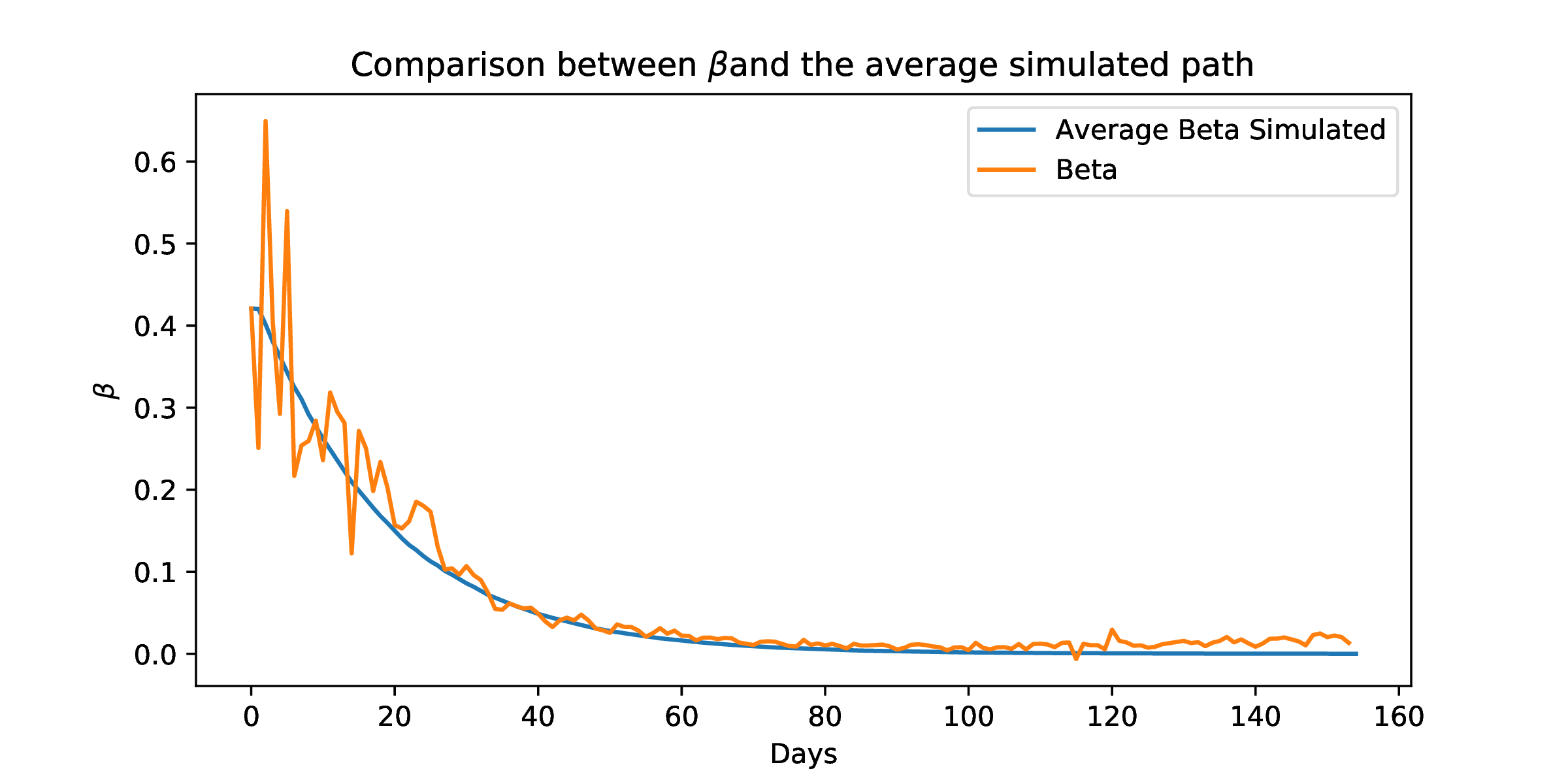}
\caption{beta}
\end{figure}

\begin{figure}[H]
\centering
\includegraphics[scale=0.18]{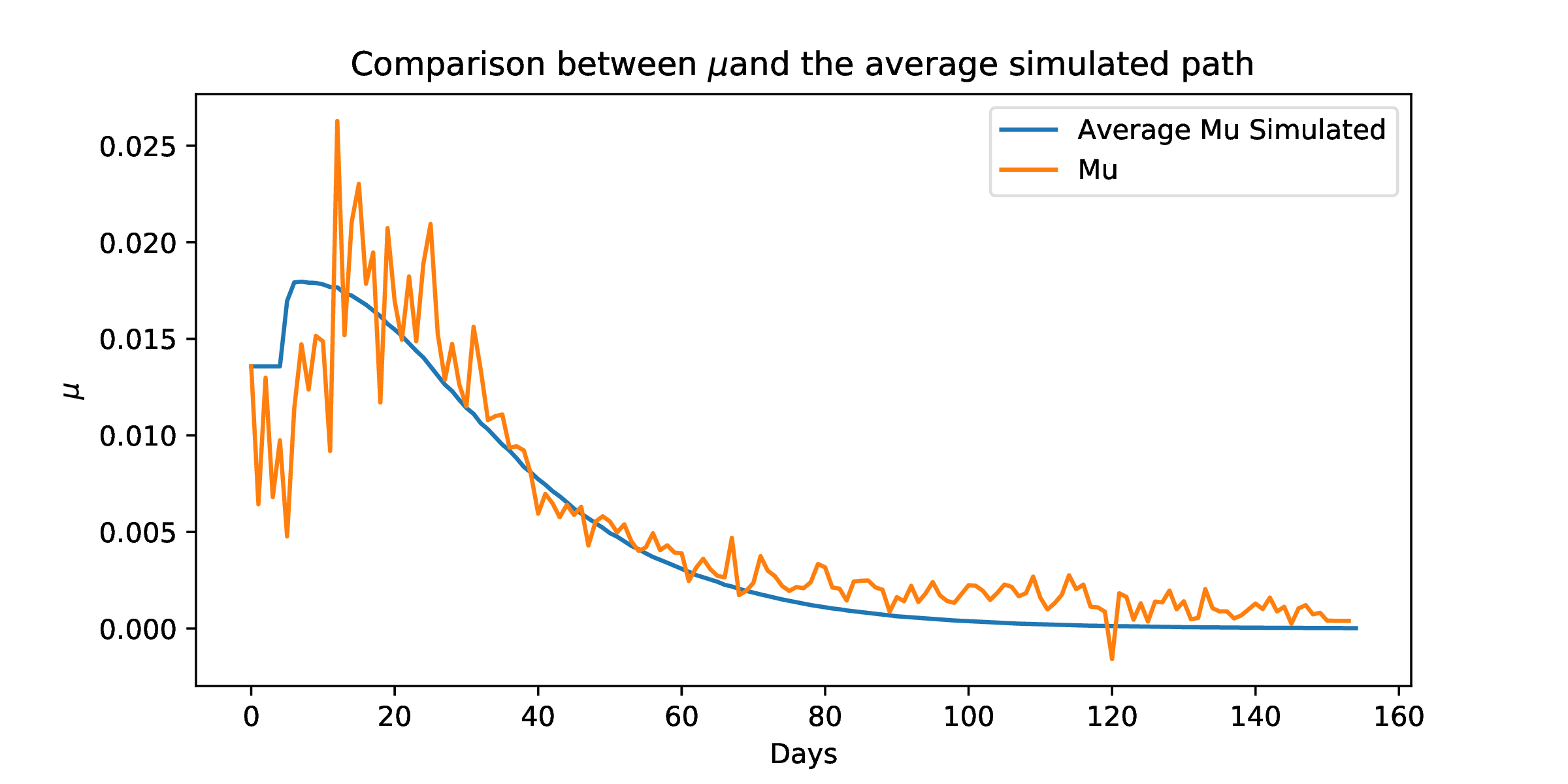}
\caption{mu}
\end{figure}

\begin{figure}[H]
\centering
\includegraphics[scale=0.18]{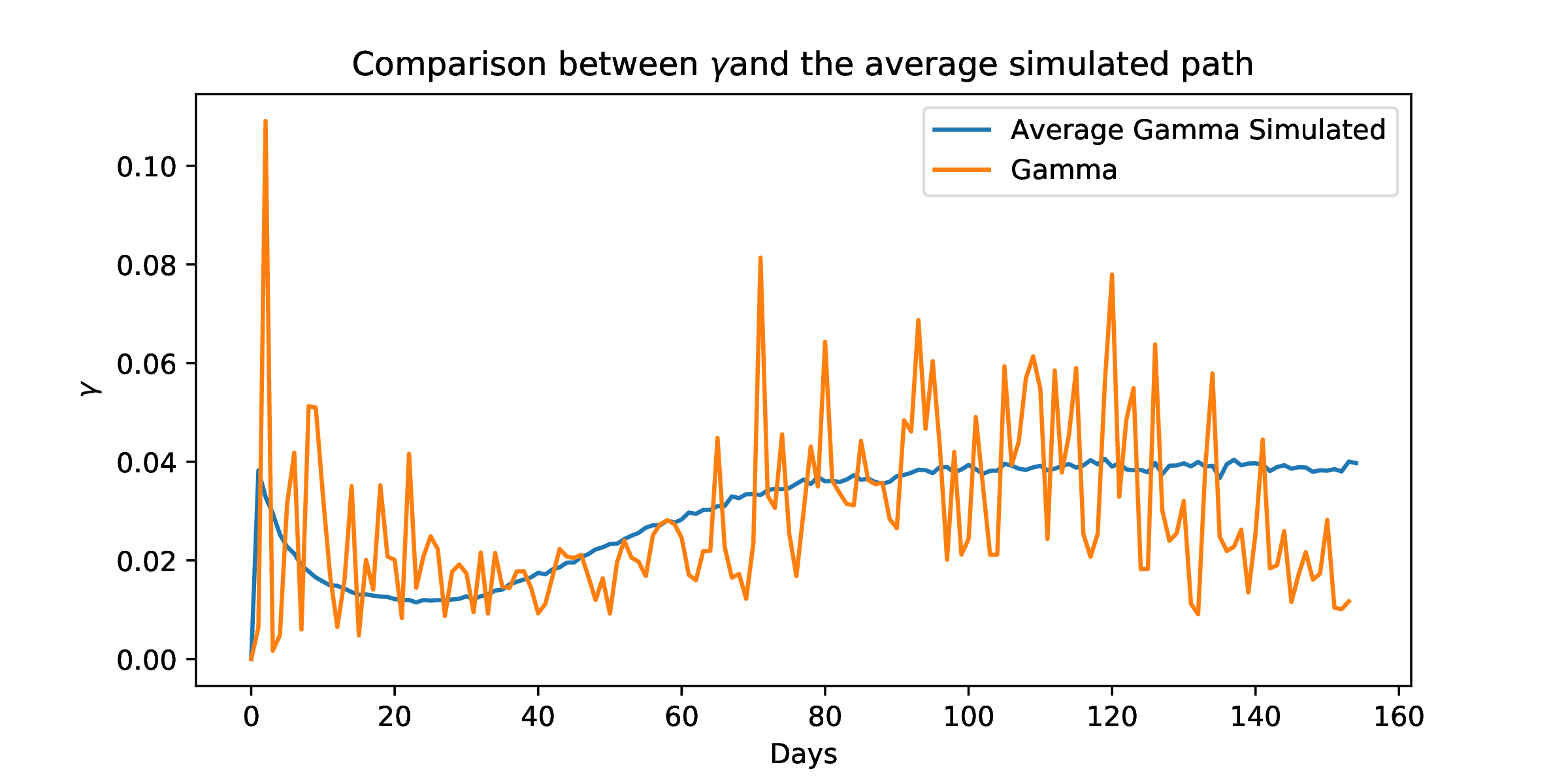}
\caption{gamma}
\end{figure}

\subsection{Some simulated paths}

Some simulated paths for $\beta, \gamma$ and $\mu$ can be observed in the following figures:

\begin{figure}[H]
\centering
\includegraphics[scale=0.18]{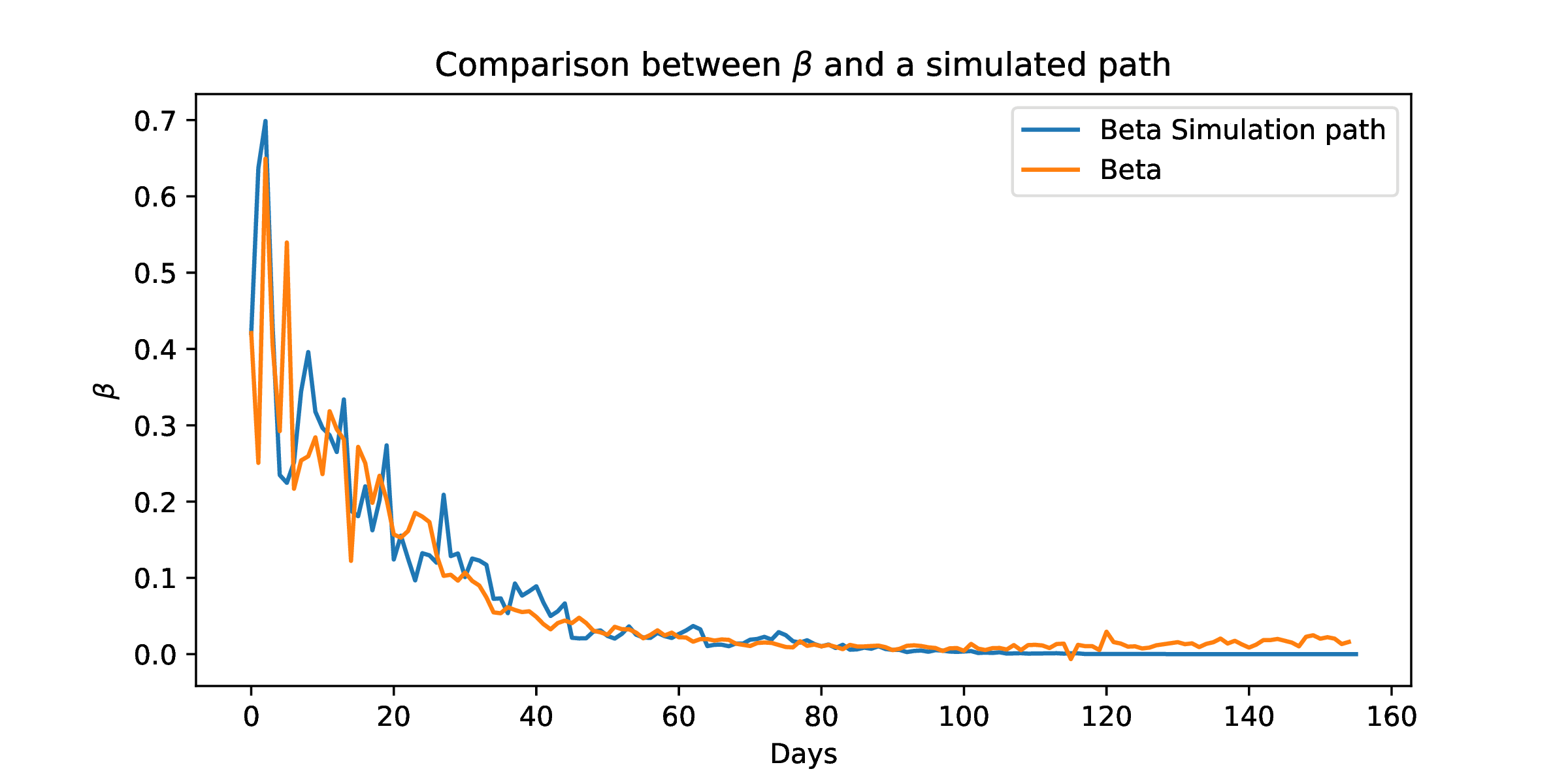}
\caption{$\beta$}
\end{figure}

\begin{figure}[H]
\centering
\includegraphics[scale=0.18]{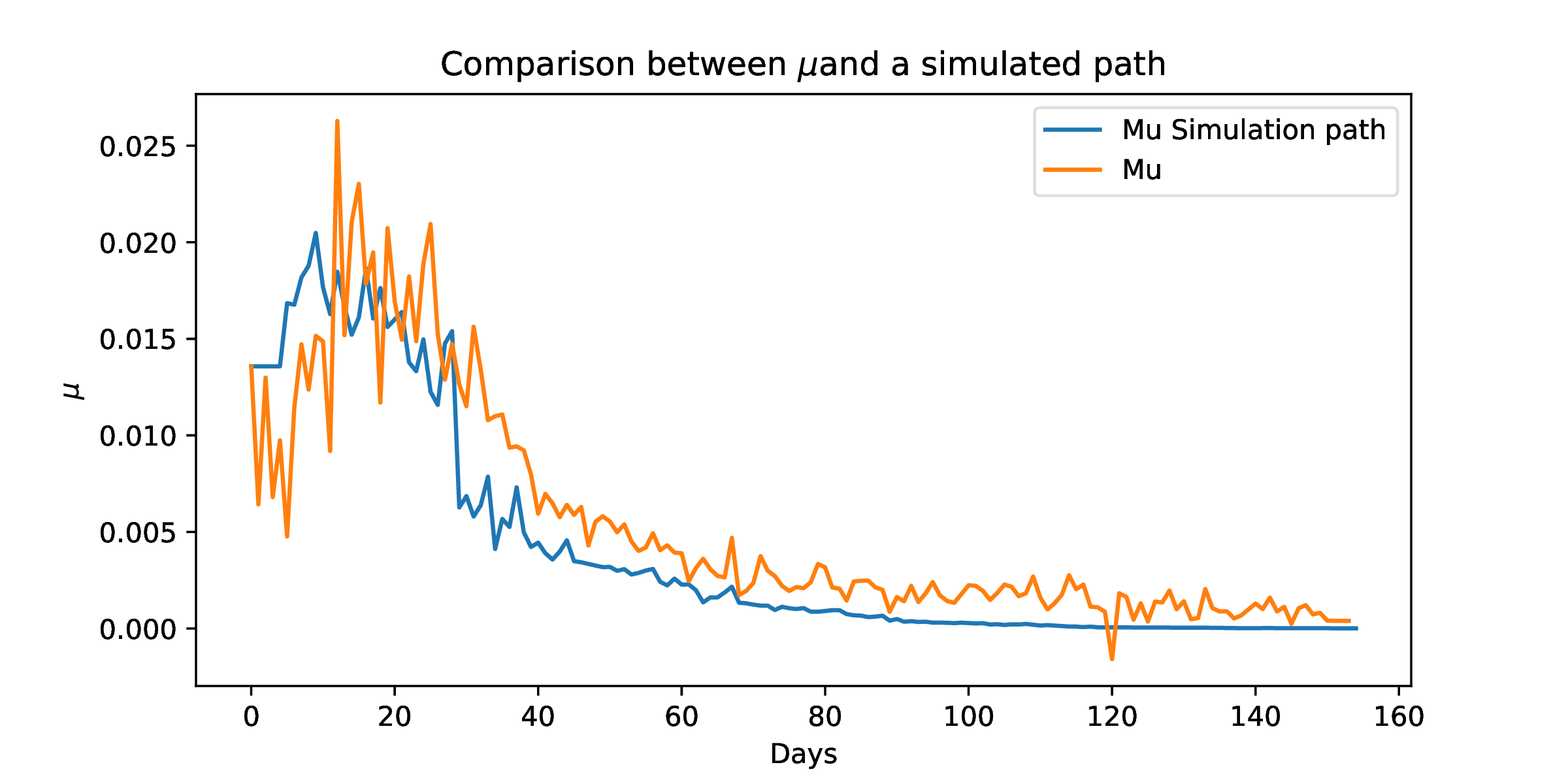}
\caption{$\mu$}
\end{figure}

\begin{figure}[H]
\centering
\includegraphics[scale=0.18]{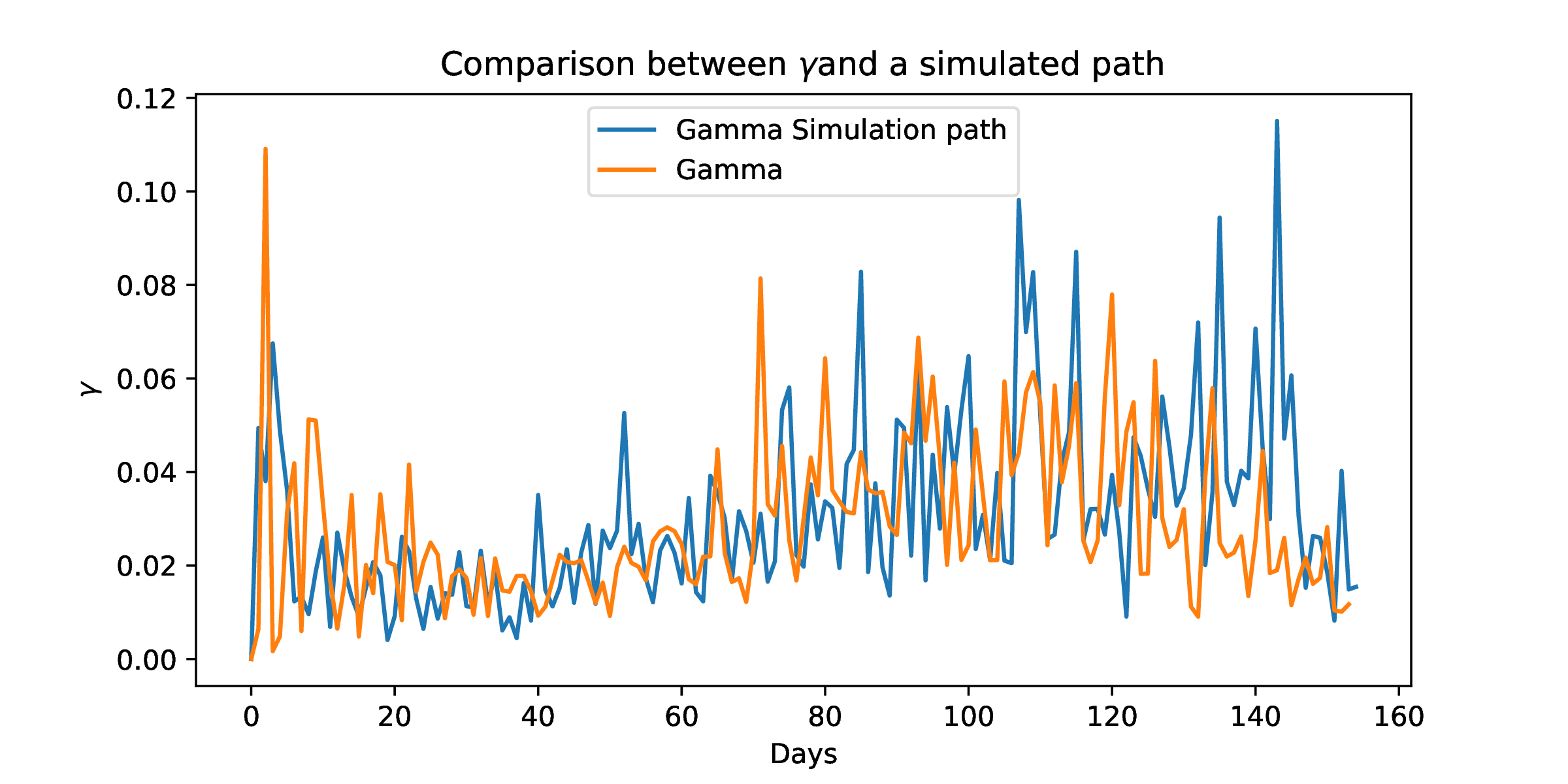}
\caption{$\gamma$}
\end{figure}

Different random scenarios can be seen in the following figures:

\begin{figure}[H]
\centering
\includegraphics[scale=0.18]{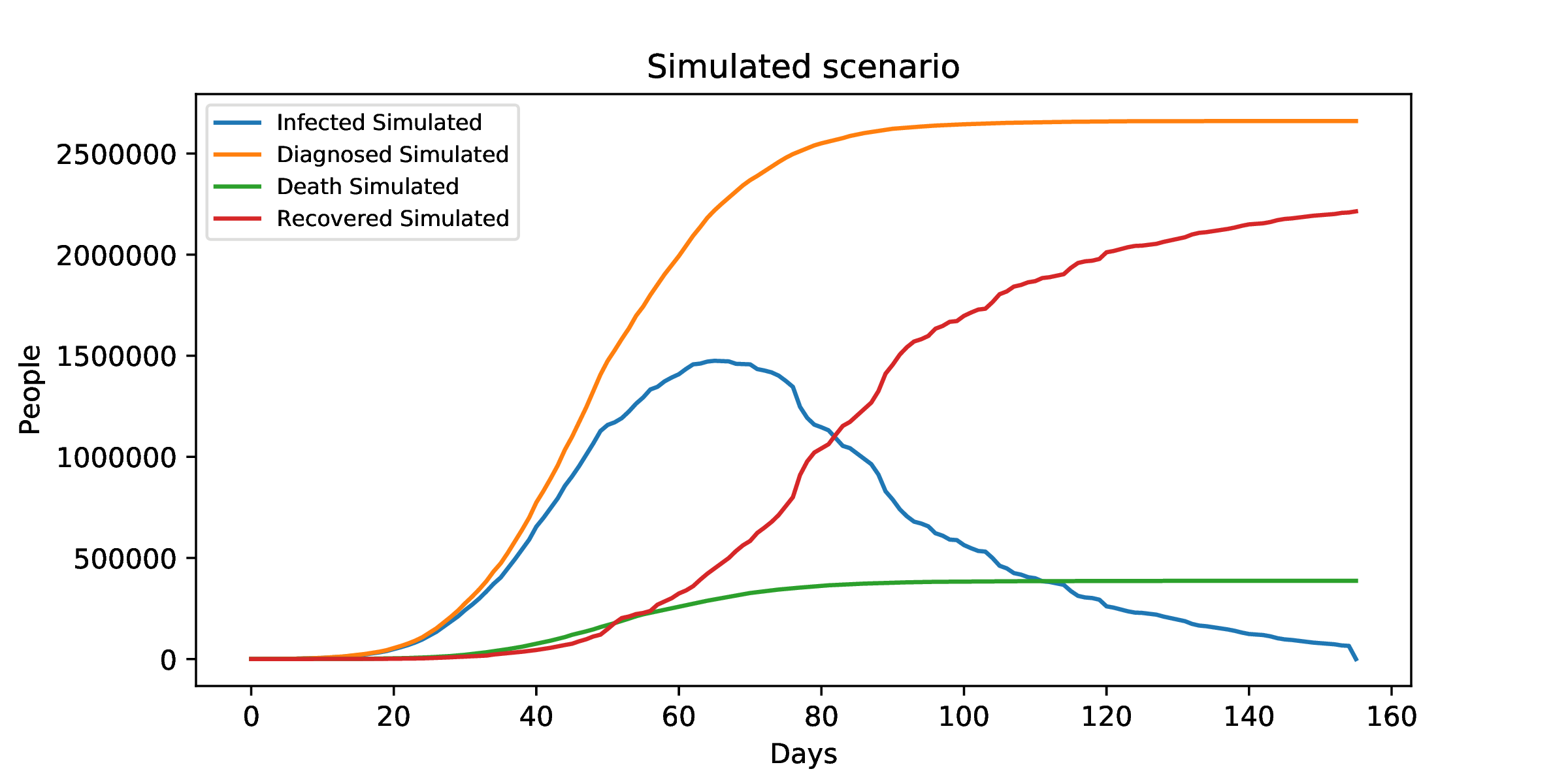}
\caption{Simulated scenario}
\end{figure}
\begin{figure}[H]
\centering
\includegraphics[scale=0.18]{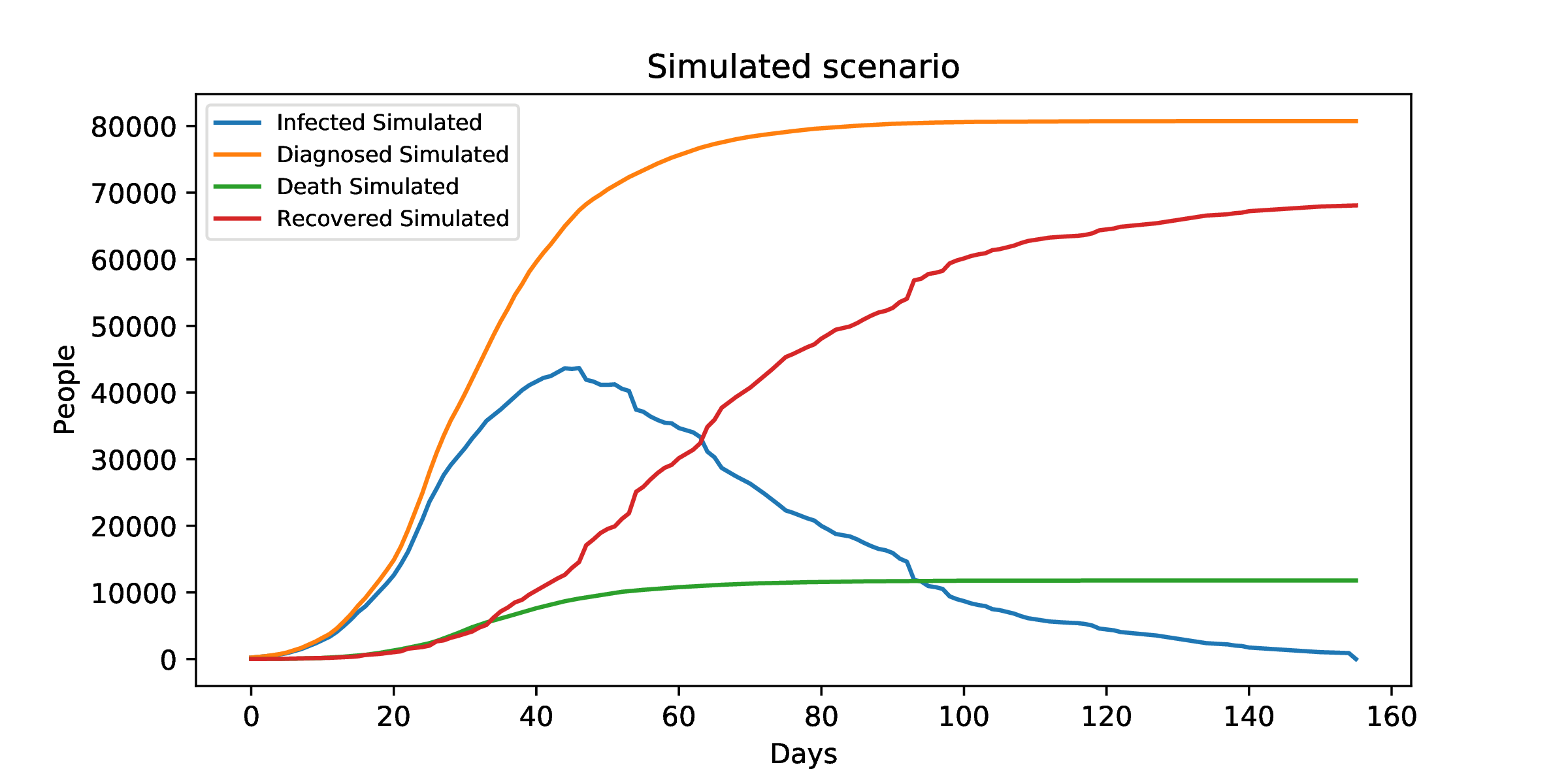}
\caption{Simulated scenario}
\end{figure}
\begin{figure}[H]
\centering
\includegraphics[scale=0.18]{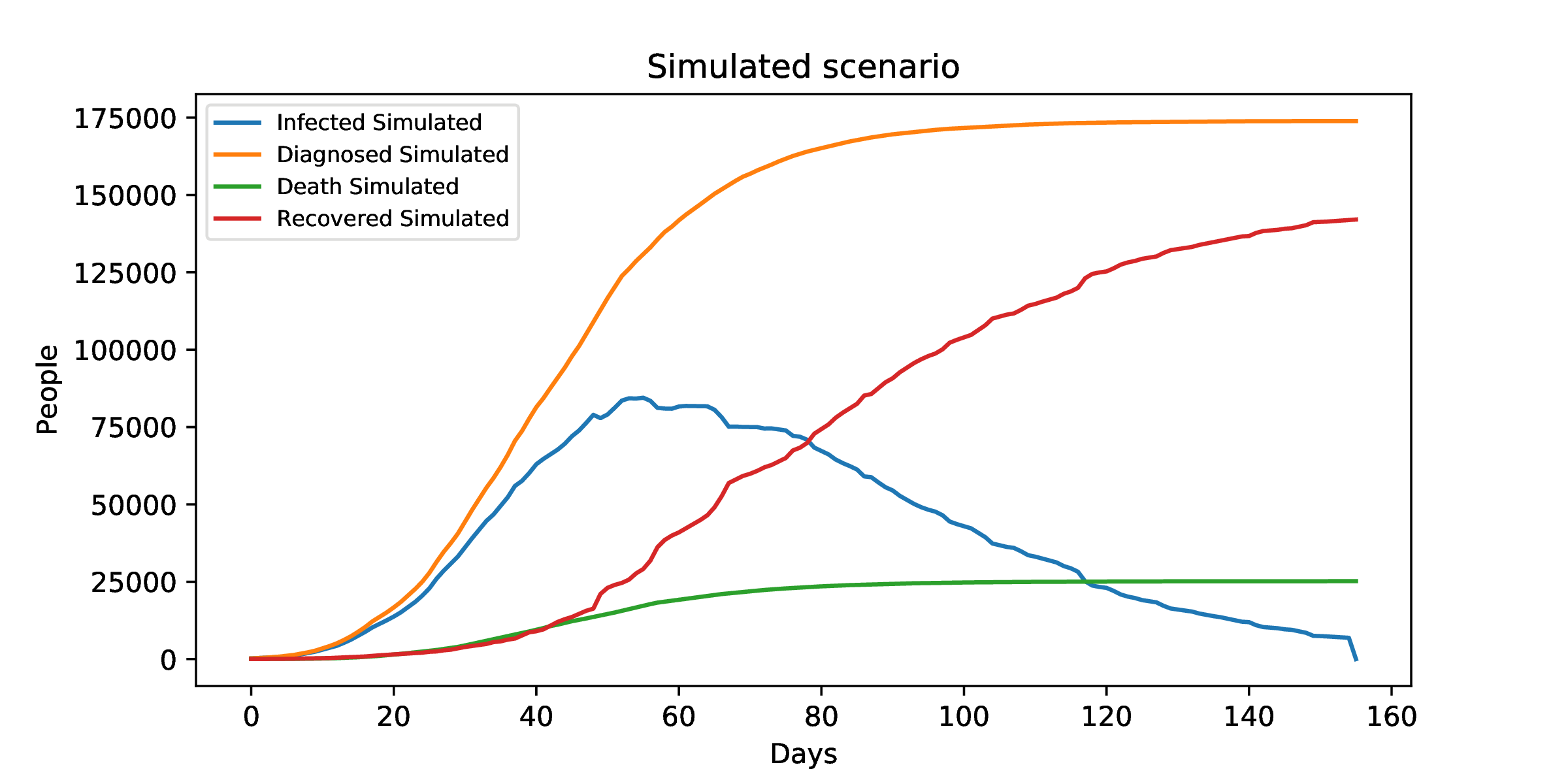}
\caption{Simulated scenario}
\end{figure}

\section{Discussion}

In the present work, we have extended the classic SIRD model including stochastic parameters and we obtain an easy-to-calibrate pure probabilistic model. We have been able to reproduce the evolution of the parameters of the Italian outbreak using only seven parameters. The properties of the fractional Brownian motion and the relationship between the parameters of the model are able to reproduce the exponential and logistic trends observed for example in \cite{C}. Therefore, we have been able to fit the empirical data as well as imitate the noise of the variables. One of the advantages of having a stochastic model is that we are capable of generating a wide variety of scenarios.

One challenging problem is now to model the evolution after a lockdown. This translates into a new dynamics of $\beta$, probably driven by a fBm with $H>\frac12$. Moreover, the comparison between different countries that applied different policies during the pandemic would be of great interest.


\end{document}